\documentclass[11pt,letterpaper]{amsart}
\usepackage{header}
\numberwithin{equation}{section}
\begin{document}
\title[Scattering from potentials near a subspace]{Scattering for Schr\"{o}dinger operators with potentials concentrated near a subspace}
\author{Adam Black}
\address{Department of Mathematics\\
  Yale University\\
 New Haven, CT 06511}
 \email[]{adam.black@yale.edu}
\author{ Tal Malinovitch}
\address{Department of Mathematics\\
  Yale University\\
 New Haven, CT 06511}
  \email[]{tal.malinovitch@yale.edu}
\maketitle
\begin{abstract}
    We study the scattering properties of Schr\"{o}dinger operators with bounded potentials concentrated near a subspace of $\bbR^d$. For such operators, we show the existence of scattering states and characterize their orthogonal complement as a set of surface states, which consists of states that are confined to the subspace (such as pure point states) and states that escape it at a sublinear rate, in a suitable sense. We provide examples of surface states for different systems including those that propagate along the subspace and those that escape the subspace arbitrarily slowly. Our proof uses a novel interpretation of the Enss method \cite{Enss} in order to obtain a dynamical characterisation of the orthogonal complement of the scattering states. 
\end{abstract}
\tableofcontents
 \clearpage
\section{Introduction}\label{intro}
In this paper, we study the scattering properties of Schr\"{o}dinger operators with potentials concentrated near a subspace of $\bbR^d$. This is one of many models of a quantum particle interacting with a surface. For such operators, we show the existence of scattering states and characterize their orthogonal complement as a set of surface states, which consists of states that are confined to the subspace (such as pure point states) and states that escape it at a sublinear rate, in a suitable sense. We provide examples of surface states for different systems including those that propagate along the subspace and those that escape the subspace arbitrarily slowly.\par
\subsection{Motivation and prior work}
Our work is motivated by the vast literature studying the scattering theory of Schr\"{o}dinger operators with potentials that decay at infinity. Typically, these are self-adjoint operators on $\calH=L^2(\bbR^d)$ of the form
\begin{align}\label{HDef}
    H=H_0+V
\end{align}
where $H_0=-\Delta$ and $V$, the potential, is a real-valued multiplication operator. For \emph{short range} potentials, that is, those with sufficiently fast decay, one is interested in showing that the \emph{wave operators}
\begin{align*}
    \Omega^{\pm}=\slim_{t\rightarrow \mp\infty}e^{itH}e^{-itH_0}
\end{align*}
exist on all of $\calH$ and are \emph{asymptotically complete} in the sense that their range is equal to the continuous subspace of $H$. Intuitively, states in the range of $\Omega^\pm$ behave like free waves as $t\rightarrow \mp\infty$, in the following sense: if $\Omega^- \psi=\varphi$, then
\begin{align*}
    \lim\limits_{t \rightarrow \infty }\|e^{-itH_0}\psi -e^{-itH}\varphi\|=0
\end{align*}
 Asymptotic completeness then means that \emph{all} states in the continuous subspace of $H$ scatter to free waves.\par
We make no attempt to comprehensively review the multitude of results concerning which assumptions on $V$ yield asymptotic completeness. However, we mention the seminal work of Agmon \cite{agmon1975spectral} (and the references therein), in which asymptotic completeness is shown for $V$ satisfying, for instance,
\begin{align*}
    V(x)=O(|x|^{-(1+\epsilon)})\text{ as }x\rightarrow\infty
\end{align*}
Our paper is based on the work on Enss \cite{Enss} showing asymptotic completeness for potentials satisfying a short range condition, which for a bounded potentials can be written as
\begin{align}\label{EnssCondition}
    \|V\chi_{B_r^c}\|\in L^1(r)
\end{align}
where $\chi$ denotes an indicator function and $B_r^c$ is the complement of the ball of radius $r$ in $\bbR^d$.\par
In a related direction, many authors have investigated the scattering theory of Schr\"{o}dinger operators with \emph{anisotropic} potentials that have different behavior in different coordinate directions (see, for example, \cite{Carmona,davies1977scattering,DaviesSimon,boutet1996some,Richard}). Building on one-dimensional results of Carmona \cite{Carmona}, Davies and Simon \cite{DaviesSimon} investigated potentials $V$ that are periodic in the coordinate directions $\{x_1,...,x_{d-1}\}$ but with different spatial asymptotics as $x_d$ goes to plus or minus infinity. They showed that in this setting, the absolutely continuous subspace of $H$ decomposes into pieces that, under the evolution of $H$, move to $\pm\infty$ in the $x_d$ coordinate and surface states that are localized near the hypersurface $\{x_d=0\}$ for all time. We review this result more thoroughly in Section \ref{ExamplesSection}, but for now we note that even if $V$ goes to $0$ rapidly as $x_d\rightarrow \pm\infty$, there may still exist surface states in the ac subspace of $H$. Furthermore, a state in the range of $\Omega^\pm$ cannot be localized near a hypersurface for all time (see Section \ref{IsoIntersection}) so the presence of ac surface states may be thought of as an obstruction to asymptotic completeness. Such states may also be seen if $V$ decays sufficiently slowly in some directions. In this case, originally studied by Yafaev \cite{yafaev1979break}, one may observe states which disperse away from the support of $V$ slower than a free wave (see Section \ref{ExamplesSection} for more details). Finally, we remark that asymptotic completeness may also fail in the sense that $\textrm{Ran}(\Omega^-)$ may no longer be equal to $\textrm{Ran}(\Omega^+)$. For $d=3$, one may observe such behavior in settings similar to those considered below \cite{davies1977scattering}. \par
In view of the circle of ideas recalled above, one may naturally ask what can be said about the scattering theory of potentials that decay at infinity \emph{but only in some coordinate directions}. By this, we mean a potential $V$ that is concentrated (in a sense to be specified later) near the surface $\{x\in\bbR^d\mid x_{k+1}=\cdots x_d=0\}$ for some $1\leq k<d$. The aforementioned class of examples shows that one cannot expect asymptotic completeness in this setting because some states may undergo transport along the surface. However, one has the following very plausible physical picture: a state which moves away from the surface as time evolves should feel the influence of the potential less and less, so it should behave asymptotically like a free particle and therefore be in the range of the wave operator. This suggests that there is a dichotomy between states that remain near the surface and those that are asymptotically free \emph{irrespective of the precise nature of $V$.} So, one should really ask: for $V$ as above, is the orthogonal complement of $\Ran{\Omega^\pm}$ given by the space of surface states? The present paper is an affirmative answer to this question.\par
Before stating our results, let us mention that many authors have studied the spectral and scattering theory of surface potentials due in part to their physical importance. We refer the reader to \cite{davies1977scattering,de2003dynamical,frank2003scattering,grinshpun1995localization,hundertmark2000spectral,JL-Cor,JL-Spectral,JL-Sur} for some idea of the questions that have been investigated for surface models. In these papers and others, the authors are usually interested in surface potentials with some additional structure. For instance, among other examples, Davies and Simon \cite{DaviesSimon} consider a partially periodic potential so that they may leverage symmetry. Other authors investigate random surface potentials \cite{de2003dynamical,grinshpun1995localization} or a (possibly discrete) half-space model with some boundary condition (such as \cite{frank2003scattering,JL-Cor,JL-Spectral,JL-Sur}). In many of these cases, additional structure allows for a better description of the surface subspace than one might hope for in full generality, either by showing it is trivial \cite{JL-Cor} or by giving a more restrictive definition \cite{DaviesSimon}. In this paper, we make significantly weaker assumptions on $V$ - only that it is bounded and has the right decay away from the surface - at the price of a more inclusive description of the surface states. Therefore, many of these prior models fall within the purview of our theorem.
\subsection{Model and results}
We consider a self-adjoint operator $H$ on $\calH=L^2(\bbR^d)$ of the form (\ref{HDef}), where $V$ is a real-valued bounded potential such that
\begin{align*}
	&\supp V\subset \{x\in \bbR^d\mid \|x^\perp\|=\|(x_{k+1},\dots,x_d)\|\leq r_0\}=:S_{r_0}^k\\
	&\sup_{x\in\bbR^d}|V(x)|=M<\infty 
\end{align*}
for some $r_0>0$ and $1\leq k< d$. Here and throughout, $\|\cdot \|$ refers to either the euclidean norm or the norm of $\calH$. Since $k$ is fixed throughout the paper, we will suppress it in the notation.
We define the space of surface states to be
\begin{align*}
    \calH_{\textrm{sur}}=\{\psi \in \calH \mid \forall v>0, \lim\limits_{t\rightarrow \infty}\|\chi_{S_{vt}}e^{-itH} \psi\|= \|\psi\|\}
\end{align*}
Our main theorem is that
\begin{theorem}\label{thms}
\leavevmode
\begin{enumerate}[label={(\roman*)},itemindent=1em]
    \item \label{existenceTheorem}(Existence) For all $\psi \in \mathcal{H}$ the limits $\Omega^\pm\psi$ exist. Furthermore, $\sigma(H_0)\subset \sigma_\mathrm{ac}(H)$.
    \item \label{Complete} (Completeness) We have 
    \begin{align*}
    \calH=\calH_\mathrm{sur}\oplus\Ran(\Omega^-)
\end{align*}
\end{enumerate}
\end{theorem}
The existence result may essentially be found in \cite{hundertmark2000spectral}, though we supply our own proof. See also Chapter 2, Section 10 of \cite{perry1983scattering} for a related existence theorem.\par
\begin{remark}
The above theorem may be easily generalized to allow $V$ satisfying
\begin{align}
	&\|\chi_{S_R^c}V\|_{\textrm{op}}\in L^1_R\label{shortrangeCond} \\
	&\sup_{x\in\bbR^d}|V(x)|=M<\infty
\end{align}
that is, potentials $V$ which decay perpendicular to the surface in a short range way. Broadly speaking, the $L^1$ condition enters in a similar way as in \cite{Enss}. For simplicity of presentation we have restricted to the case where $\chi_{S_{R}}^cV$ is in fact $0$ for $R$ large enough, but we have explained how to adapt our proof to this generalization in Appendix \ref{appendix:Decay}.
\end{remark}
\begin{remark}
The definition of $\calH_\textrm{sur}$ is closely related to the notion of a \emph{minimal velocity estimate} as exhibited in \cite{hunziker1999minimal}. A typical estimate of this type for a state $\psi$ might be of the form
\begin{align*}
    \|\chi_{B_{vt}}e^{-itH}\psi\|\leq Ct^{-\ell}\|\psi\|
\end{align*}
for some $\ell>0$ and all $v$ less than some $v_0$. Such an estimate usually results from a Mourre estimate on some energy interval, in the presence of which one already expects asymptotic completeness (for a self-contained exposition of these ideas, see Chapter 4 of \cite{derezinski1997scattering}). An easy corollary of our Theorem \ref{thms} is that a state $\psi$ is a scattering state if it satisfies a minimal velocity estimate relative to the region $S_{vt}$, i.e., if for some $v>0$
\begin{align*}
    \liminf\limits_{t\rightarrow \infty}\|\chi_{S_{vt}}e^{-itH}\psi\|=0
\end{align*}
Thus, our theorem provides a dynamical criterion for asymptotic completeness, which may be verified via commutator methods. 
\end{remark}

\subsection{Methodology: the Enss Method of scattering}
We rely on the Enss method of scattering originally developed in \cite{Enss}, whose geometric flavor is well-suited to our problem. The Enss method realizes the physical intuition developed above: if $V$ satisfies (\ref{EnssCondition}), a state which moves away from the origin under the $H$ evolution is asymptotically free.
In Enns' original argument, one fixes a state $\psi$ in the absolutely continuous subspace of $H$ and finds a sequence of times $t_n\rightarrow\infty$ for which $\psi_n=e^{-it_nH}\psi$ satisfies
\begin{align*}
    \|\chi_{B_n}\psi_n\|\rightarrow 0
\end{align*}
so that $\psi_n$ is moving away from the origin. This is possible for $V$ a relatively bounded perturbation of $H_0$ with relative bound less than $1$ by the celebrated RAGE theorem \cite{AmreinGeorgescu,Ruelle}, which says that a state $\psi$ in the continuous subspace escapes every compact set $K$ in a time mean sense:
\begin{align*}
    \lim_{T\rightarrow\infty}\frac{1}{T}\int_0^T\|\chi_K\psi_t\|\,dt=0
\end{align*}
Along the sequence $\{t_n\}_{n=1}^\infty$, one then performs a \emph{phase space decomposition} of $\psi_n$ into incoming and outgoing pieces:
\begin{align*}
    \psi_n=\psi_{n,\textrm{in}}+\psi_{n,\textrm{out}}+o(1)
\end{align*}
Both $\psi_{n,\textrm{in}}$ and $\psi_{n,\textrm{out}}$ are spatially localized far from the origin with momenta that point roughly toward or away from the origin respectively. These phase space properties of $\psi_{n,\textrm{in/out}}$ guarantee that
\begin{align*}
&\lim_{n\rightarrow\infty} \|(\Omega^--\id)\psi_{n,\textrm{out}}\|=0\\
&\lim_{n\rightarrow\infty} \|(\Omega^+-\id)\psi_{n,\textrm{in}}\|=0
\end{align*}
from which asymptotic completeness is an easy consequence.\par
In trying to apply the above outline to our setting verbatim, one encounters the problem that one cannot use the RAGE theorem to see that a continuous state moves away from the surface, as the surface is not compact. To proceed, we provide a novel interpretation of Enss original argument that does not rely on any a priori properties of the continuous subspace. Working in the original Enss setting, we fix a state $\psi$ orthogonal to $\Ran(\Omega^-)$ and perform a phase space decomposition along an arbitrary time sequence increasing to infinity, now keeping the piece of $\psi$ close to the origin (in the above, this piece was $o(1)$ by the RAGE theorem):
\begin{align*}
    \psi_n=\psi_{n,\textrm{bounded}}+\psi_{n,\textrm{in}}+\psi_{n,\textrm{out}}+o(1)
\end{align*}
Here, $\psi_n=e^{-it_nH}\psi$ as before and $\psi_{n,\textrm{bounded}}$ is essentially $\chi_{B_n}\psi_n$. One can argue that $\psi_{n,\textrm{in}}$ goes to $0$ as $n\rightarrow\infty$ and the fact that $\psi\perp \Ran(\Omega^{-})$ implies the same for $\psi_{n,\textrm{out}}$. Thus, $\psi_n$ is asymptotically equal to $\psi_{n,\textrm{bounded}}$ and by varying over all time sequences one may show that
\begin{align}\label{BoundEq}
    \lim\limits_{n\rightarrow\infty}\liminf_{t\rightarrow\infty} \|\chi_{B_n}\psi_t\|=\|\psi\|
\end{align}
In other words, Enss' argument provides a geometrical characterization of the orthogonal complement of $\Ran(\Omega^-)$ as the set of bound states. Indeed, it is a consequence of the RAGE theorem that the states satisfying (\ref{BoundEq}) are precisely the pure point states of $H$, but one need not know this to obtain this interesting theorem.\par
Our adaptation of this argument to surface scattering will require that the operators implementing the phase space decomposition have better monotonicity properties than those originally used by Enss. To this end, we adopt Davies' \cite{davies1980enss} point of view on the Enss' method by defining families of phase space observables. This formulation allows us to define the decomposition in a natural way, via operators which are almost projections onto subsets of phase space. Choosing these operators in the correct way allows us to study the evolution in a lower dimensional space, i.e. only in the directions  perpendicular to the surface. For the reader's convenience, we have collected various results about these observables in Appendix \ref{DaviesProperties}. This is particularly important because throughout the proof we will use a phase space characterisation of the surface states. The precise definition of this characterisation will be given in Section \ref{HsurTildeDef}, but for now it can be described as consisting of states that either evolve close to the surface or propagate away from the surface with momenta roughly parallel to the surface.
\begin{remark}
 A natural question that arises from these two characterisations of $\calH_\textrm{sur}$ is: can there truly be surface states that propagate away from the subspace? If so, these states would have to do so at a sublinear rate and with highly restricted momenta. Indeed, following \cite{DaviesSimon}, one may define
\begin{align*}
    \calH_\textrm{sur}'(H)=\{\psi\mid\lim_{R\rightarrow\infty} \sup_{t\geq 0} \|\chi_{S_R^c}e^{-itH}\psi\|=0\}
\end{align*}
which contains all states that evolve close to the subspace. This definition will be convenient to work with in Section \ref{ExamplesSection}. As shown in Proposition \ref{Hsur'}, $\calH_\textrm{sur}'\subset \calH_\textrm{sur}$, so we may reformulate our question as: is there some choice of potential $V$ so that $\calH_{\textrm{sur}}\setminus \calH_\textrm{sur}'(H)$ is non-empty?\par
Indeed, such potentials do exist: following Yafaev \cite{yafaev1979break}, in Section \ref{ExamplesSection} we show that $V$ decaying like a long range potential in the $x^\parallel$ direction may produce such states. However, we will show in Section \ref{ExamplesSection} that at least for $V$ partially periodic or $V$ that decays to a limit at $\infty$ quickly enough, $\calH_\textrm{sur}'= \calH_\textrm{sur}$.
\end{remark}

\addtocontents{toc}{\setcounter{tocdepth}{-10}}
\subsection*{Outline of paper} In Section \ref{Def} we provide some notation as well as define $\tilde{\calH}_{\textrm{sur}}$, the auxiliary surface subspace that will be used in the proof of Theorem \ref{thms} extensively. In Section \ref{Existence}, we prove \ref{existenceTheorem} of Theorem \ref{thms}, in other words the existence of scattering states. In Section \ref{StateDecom}, we develop the Enns decomposition (Theorem \ref{thm:decomposition}) for our setting, stated using the phase space observables of Davies. The decomposition is proved, as in the original Enss paper \cite{Enss}, by combining Cook's method with several applications of non-stationary phase. This decomposition is the main ingredient used to show, in Section \ref{Spans}, that $\tilde{\calH}_{\textrm{sur}}$ and $\Ran(\Omega^\pm)$ span all of $\calH$, as described in the sketch above. In Section \ref{Intersection} we show that the intersection of these two subspaces is trivial, yielding our first completeness result (Lemma \ref{SemiComplete}). For this, we show that the intersection is unitarily equivalent to $\tilde{\calH}_{\textrm{sur}}(H_0)$, the surface states of the free evolution, which we show to be trivial by a direct computation. Then, we again use the method of non-stationary phase to give a better characterization of the surface states, namely to show that $\tilde{\calH}_{\textrm{sur}}$ is in fact equal to $\calH_{\textrm{sur}}$. In Section \ref{ExamplesSection} we consider some special classes of potentials and discuss their surface states, relating them to known results where relevant. Finally, in Appendix \ref{appendix:Decay}, we explain how to accommodate short range decay of the potential away from the surface.
\subsection*{Acknowledgment} We are grateful to our advisor, Wilhelm Schlag, for leading us towards this problem, and for his guidance and encouragement during this work. We also thank Michael Weinstein and Amir Sagiv for discussions that improved the definition of $\calH_\textrm{sur}$.
\addtocontents{toc}{\setcounter{tocdepth}{2}}

\section{Definitions and Results}\label{Def}
\subsection{Notation and Conventions}
For any $\ell>0$ we use the following
\begin{itemize}
    \item We let $\calH$ denote $L^2(\bbR^d)$ with norm $\|\cdot\|$ and use the convention that its inner product $\braket{\cdot,\cdot}$ is anti-linear in the first argument and linear in the second.
    \item The symbols $\|\cdot\|$ and $\braket{\cdot,\cdot}$ will also be used for the norm and inner product on $\bbR^\ell$.
    \item $d(\cdot,\cdot)$ is used for the distance between points or subsets of $\bbR^\ell$.
    \item $B_r$ will mean the ball of radius $r$ centered at the origin in either $\bbR^\ell$ or $\calH$ depending on context.
    \item For $A\subset \bbR^\ell$, $A^c$ denotes its complement.
    \item $\chi_A$ will mean the indicator function of $A\subset \bbR^\ell$.
    \item $A\Subset B$ denotes that $A$ is compactly contained in $B$.
    \item $\calS=\calS(\bbR^d)$, the Schwartz space. 
    \item We use the following convention for the Fourier transform of $f\in\calH$:
    \begin{align*}
        &\hat{f}(\xi)=\calF(f)(\xi)=(2\pi)^{-\frac{d}{2}}\int\limits_{\bbR^d}f(x)e^{-ix\xi} \,dx\\
        &\calF^{-1}(\hat{f})(x)=(2\pi)^{-\frac{d}{2}}\int\limits_{\bbR^d}\hat{f}(\xi)e^{ix\xi} \,d\xi
    \end{align*}
    \item For $x=(x_1,\dots,x_d)\in \bbR^d=\bbR^k\times \bbR^{d-k}$ we will often write $x^\parallel= (x_1,\dots,x_k)$ and $x^\perp=(x_{k+1},\dots,x_d) \in \bbR^{d-k}$ for $k$ some integer $1\leq k \leq d-1$. We will refer to the $\bbR^k$ components as longitudinal and the $\bbR^{d-k}$ components as transverse.
    \item $S_R\subset \bbR^d$ is the set of points within $R$ of $\bbR^k\times \{0\}$:
    \begin{align*}
        S_R=\{x\in \bbR^d\mid \|x^\perp\|\leq R\}
    \end{align*}
    \item For $\alpha>0$ define the following family of subspaces of $\calH$
    \begin{align*}
        \calD_{\alpha}=\textrm{Span}(\{\psi_i^\parallel\otimes \psi_i^\perp \mid \psi^\parallel_i\in L^2(\bbR^k), \psi ^\perp_i \in \calS(\bbR^{d-k}), \supp \widehat{\psi^\perp_i} \Subset B_\alpha^c\})
    \end{align*}
    \item For the definitions of $P_{\delta}(E)$ and $\hat{\eta}_{x,p;\delta}$ see Section \ref{HsurTildeDef} below. 
\end{itemize}
\subsection{Definition of the auxiliary surface subspace}\label{HsurTildeDef}
As mentioned above, for the proof of part \ref{Complete} of Theorem \ref{thms}, asymptotic completeness, it will be more convenient to work with a different subspace, denoted $\tilde{\calH}_{\textrm{sur}}$. We will show in Section \ref{SurSpaceChar} that it is in fact equal to $\calH_{\textrm{sur}}$.
The definition of this subspace and the arguments that follow depend crucially on the ability to localize a state into a subset of phase space. For this, we will follow the formulation of phase space observables developed in \cite{davies1976quantum}. 
To this end, choose $\eta\in\calS(\bbR^d)$, such that $\|\eta\|=1$ and $\supp \hat{\eta}\subset B_1$. Let $\eta_\delta$ be such that $\hat{\eta}_\delta(p)=\delta^{-\frac{d}{2}}\hat{\eta}(\frac{p}{\delta})$, a rescaling of $\eta$, so that $\supp \hat{\eta}^d_\delta\subset B_\delta$ and $\|\eta_{\delta}\|=1$.\par
Now define the following family of coherent states by translating $\eta_\delta$ in phase space:
\begin{align*}
    &\hat{\eta}_{x,p;\delta}(\xi)=e^{-ix\xi}\hat{\eta}_\delta(\xi-p)
\end{align*}
or equivalently
\begin{align*}
    &\eta_{x,p;\delta}(y)=e^{ip(y-x)}\eta_\delta(y-x)
\end{align*}
We use this to define a family, depending on $\delta>0$, of positive-operator-valued measures as in \cite{davies1980enss}, which serve as phase space observables. For any $E\subset \bbR^{2d}$ Borel and $\psi \in \calH$ let
\begin{align*}
    P_\delta(E)\psi=(2\pi)^{-d}\iint\limits_E \braket{\eta_{x,p;\delta},\psi} \eta_{x,p;\delta} \,dx\,dp
\end{align*}
which is a weakly convergent integral. These operators are closely related to the Fourier-Bargmann transform 
$\mathscr{F}_{\eta_\delta}:L^2(\bbR^d)\rightarrow L^2(\bbR^{2d})$ defined, for instance, in \cite{combescure2012coherent} Section 1.3.3. In our notation, $\mathscr{F}_{\eta_\delta}$ may written as
\begin{align*}
    (\mathscr{F}_{\eta_\delta} \psi)(x,p)= (2\pi)^{-\frac{d}{2}}\braket{\eta_{x,p;\delta},\psi}
\end{align*}
Using this, we can write $P_\delta(E)$ as
\begin{align*}
    P_\delta(E)\psi= (2\pi)^{-\frac{d}{2}}\iint\limits_E (\mathscr{F}_{\eta_\delta} \psi)(x,p) \eta_{x,p;\delta} \,dx\,dp=\mathscr{F}_{\eta_\delta}^{*}\chi_E\mathscr{F}_{\eta_\delta}\psi
\end{align*}
where $\mathscr{F}_{\eta_\delta}^{*}$ is the adjoint of $\mathscr{F}_{\eta_\delta}$. Note that $P_\delta(E)$ is self-adjoint and non-negative by construction. See \cite{davies1976quantum} for more details about the basic properties of these positive-operator-valued measures.\par
In this paper, we will choose $\eta$ that factors into functions of $x^\parallel$ and $x^\perp$:
\begin{align*}
    \eta=\eta^\parallel\otimes \eta^\perp
\end{align*}
where $\eta^\parallel \in \calS(\bbR^k)$ and $ \eta^\perp \in \calS(\bbR^{d-k})$. From now on, we will label the coordinates of  $\bbR^{2d}$ as $(x^\parallel,p^\parallel,x^\perp,p^\perp)$ where $(x^\parallel,p^\parallel)\in \bbR^k\times \bbR^k$ and  $(x^\perp,p^\perp)\in  \bbR^{d-k}\times \bbR^{d-k}$.  For $E^\parallel\subset \bbR^{2k}, E^\perp\subset \bbR^{2(d-k)}$, we can write
\begin{align*}
    P_\delta(E^\parallel\times E^\perp)=P^\parallel_\delta(E^\parallel)\otimes P^\perp_\delta(E^\perp)
\end{align*}
(see Proposition \ref{TensorClaim}).\par For $n>0$ and $m>0$, we define the far set in phase space to have space coordinates in $S_n^c$ (that is, $x^\perp \in B_n^c$) and momentum in $S_m^c$ (that is, $p^\perp \in B_m^c$), as well as its complement, the surface set:
\begin{align*}
    &W_{n,m;\textrm{far}}=  \bbR^{2k} \times (B_n^c\times B_m^c)\\
    &W_{n,m;\textrm{sur}}=(W_{n,m;\textrm{far}})^c=  \bbR^{2k} \times (B_n \times \bbR^{d-k}) \sqcup \bbR^{2k} \times (B_n^c \times B_m)
\end{align*}
In words, $W_{n,m;\textrm{far}}$ consists of states that have transverse position and transverse momentum bounded away from $0$ and $W_{n,m;\textrm{sur}}$ is its complement.
Here and elsewhere, the dimension of $B_n$ is understood from context.\par
Let $\calN_H^{m}: \calH\rightarrow \bbR^+$ denote the family of continuous seminorms
\begin{align*}
    \calN_H^{m}(\psi)=\limsup\limits_{\delta \rightarrow 0 }\limsup\limits_{n\rightarrow\infty} \sup\limits_{t\geq 0 }\|P_\delta(W_{n,m;\textrm{far}})e^{-itH} \psi\|
\end{align*}
This allows us to define the set of surface states as
\begin{align*}
    \tilde{\calH}_{\textrm{sur}}(H)=\bigcap_{m>0}\{\psi\in\calH\mid
     \calN_H^{m}(\psi) =0\}
\end{align*}
which is manifestly a closed subspace.
The expression $\tilde{\calH}_{\textrm{sur}}$ without an operator will be used throughout to denote $\tilde{\calH}_{\textrm{sur}}(H)$.

\section{Existence of the Wave Operators}\label{Existence}
To begin, we use the following direct application of the Corollary to Theorem XI.14 from \cite{RSVol3}:
\begin{lemma}\label{simpleNonstationary}
	Let $u$ be a Schwartz function such that $ \hat{u} $ has compact support. Let $ \mathcal{G}$ be an open set containing the compact set $\{2\xi \mid \xi\in\supp \hat{u}\}$. Then for any $\ell\in \bbN $, there is a constant $ C>0$ depending on $ \ell,u,$ and $ \mathcal{G} $ so that
	\begin{align*}
		|e^{-itH_0}u(x)|\leq C(1+\|x\|+|t|)^{-\ell}
	\end{align*}
	for all pairs $ (x,t) $ such that $ \frac{x}{t}\not\in \mathcal{G} $.
\end{lemma}
This is already enough to prove the existence of the wave operators:
\begin{proof}[Proof of part \ref{existenceTheorem} of Theorem \ref{thms}]
 By Cook's method (see \cite{RSVol3} 
 Theorem XI.4), it suffices to show that for $ \mathcal{D} $ a dense set in $\calH $
\begin{align*}
	\forall \psi \in \mathcal{D}, \int\limits_0^\infty \|Ve^{-itH_0}\psi\|\,dt<\infty
\end{align*}
To this end, for $\alpha>0$ define
\begin{align*}
    \calD_{\alpha}=\textrm{Span}(\{\psi_i^\parallel\otimes \psi_i^\perp \mid \psi^\parallel_i\in L^2(\bbR^k), \psi ^\perp_i \in \calS(\bbR^{d-k}), \supp \widehat{\psi^\perp_i} \Subset B_\alpha^c\})
\end{align*}
Here, $\textrm{Span}$ means \emph{finite} linear combinations so that $\bigcup_{\alpha>0}\calD_\alpha$ is dense in $L^2(\bbR^d)$.\par
By linearity, it suffices to show the existence of $\Omega^\pm$ for simple tensors in $\calD_\alpha$:
\begin{align*}
    \psi=\psi^\parallel\otimes \psi^\perp
\end{align*}
 By factoring $\chi_{S_{r_0}}=\Id\otimes B_{r_0}$ and $e^{-itH_0}=e^{-itH_0^\parallel}\otimes e^{-itH_0^\perp}$, we may write
\begin{align}\label{ext:PotDecay}
\begin{split}
    \|Ve^{-itH_0}\psi\|&=\|V\chi_{S_{r_0}}e^{-itH_0}\psi\|\leq M\|\chi_{S_{r_0}}e^{-itH_0}\psi\|=M\|e^{-itH_0^\parallel}\psi^\parallel\|\|\chi_{B_{r_0}}e^{-itH_0^\perp}\psi^\perp\|\\
    &=M\|\psi^\parallel\|\|\chi_{B_{r_0}}e^{-itH_0^\perp}\psi^\perp\|
\end{split}
\end{align}
We now estimate this last expression via Lemma \ref{simpleNonstationary}. For this, note that we have
\begin{align*}
	\{2\xi\mid \xi\in\supp\hat{\psi}_i^\perp\} \Subset  B_{2\alpha}^c
\end{align*}
Thus, if $t>\frac{r_0}{2\alpha}$ and $x\in B_{r_0}$ we have that
\begin{align*}
    \|\frac{x}{t}\|<\frac{r_0}{t}<2\alpha
\end{align*}
 Therefore, we may apply Lemma \ref{simpleNonstationary}, to see that for any $\ell>0$
\begin{align*}
	|e^{-itH_0^\perp}\psi_i^\perp(x)|\leq C(1+\|x\|+|t|)^{-\ell}
\end{align*}
for all $x\in B_{r_0}$ and $t>\frac{r_0}{2\alpha}$ where $C$ is independent of $x$ and $t$. 
Choosing $\ell$ large enough, we get that for all $t>\frac{r_0}{2\alpha}$
\begin{align} \label{ExistEq}
\begin{split}
	\|Ve^{-itH_0}\psi\|^2&\leq C  \int\limits_{B_{r_0}}(1+\|x\|+t)^{-\ell}\,dx\leq C (1+t)^{-\ell+d}
	\end{split} 
\end{align}
where $C$ denotes a constant which may change from line to line but is always independent of $x$ and $t$.
It follows immediately that
\begin{align*}
    \int\limits_0^\infty \|Ve^{-itH_0}\psi\|\,dt<\infty
\end{align*}
so that by Cook's method $\Omega^-\psi$ exists. Since $\bigcup\limits_{\alpha>0}\calD_{\alpha}$ is dense in $\calH$, we conclude that $\Omega^-\psi$ exists for all $\psi \in \calH$ and the claim for $\Omega^+$ follows from a similar argument.\par
The inclusion $\sigma(H_0)\subset \sigma_{\textrm{ac}}(H)$ is a result of the intertwining property of $\Omega^{\pm}$: 
$ \Omega^\pm $ defines a unitary equivalence between $H_0$ and $H\vert_{\Omega^\pm(\calH)}$ and $\sigma(H_0)$ is purely absolutely continuous.
\end{proof}
\section{Enss Decomposition}\label{StateDecom}
We fix $m>0$  in order to prove the following decomposition lemma. Since $m$ is  fixed in this lemma and its proof, we will often suppress it in the notation. However, it should be noted that the decomposition does depend on $m$.
\begin{theorem}\label{thm:decomposition}
Let $\{\varphi\}_{n=0}^\infty\subset \calH$ be a sequence of unit vectors. Then for any $m>0$, there exists some $\delta_0=\delta_0(m)$, so that for all $\delta\in (0,\delta_0)$ we may write
\begin{align*}
    \varphi_n=\varphi_{n;\mathrm{out}}+\varphi_{n;\mathrm{in}}+\varphi_{n;\mathrm{sur}}
\end{align*}
where these summands satisfy
\begin{align}
    &\lim\limits_{n\rightarrow \infty}\|(\Omega^--\id)\varphi_{n;\mathrm{out}}\|=\lim\limits_{n\rightarrow \infty}\|(\Omega^+-\id)\varphi_{n;\mathrm{in}}\|=0\tag{a}\label{eq:a}\\
    &\varphi_{n;\mathrm{sur}}=P_{\delta }(W_{ n,m;\mathrm{sur}})\varphi_n, P_{\delta }(W_{ n,m;\mathrm{far}})\varphi_n= \varphi_{n;\mathrm{out}}+\varphi_{n;\mathrm{in}}\tag{b}\label{eq:b}
\end{align}
If additionally $\varphi_n=e^{-it_nH}\varphi$ for some sequence of positive times $\{t_n\}_{n=0}^\infty$ then
\begin{align*}
     &\lim\limits_{n\rightarrow\infty}\|\varphi_{n;\mathrm{in}}\|=0\tag{c}\label{eq:c}\\
\end{align*}
\end{theorem}
\begin{proof}[Proof of Theorem \ref{thm:decomposition}]
We now define subsets of $\bbR^{2d}$ that decompose $W_{n,m;\textrm{far}}$ into subsets of phase space with momenta pointing towards and away from $\supp V$. For a point $(x,p)$ in phase space, this means that its transverse position and transverse momenta are either aligned or unaligned respectively:
\begin{align*}
    &W_{n,m;\textrm{out}}=\{ (x^\parallel,p^\parallel,x^\perp,p^\perp)\in W_{n,m;\textrm{far}} \mid \braket{x^\perp,p^\perp}\geq 0\}\\
    &W_{n,m;\textrm{in}}=\{(x^\parallel,p^\parallel,x^\perp,p^\perp)\in W_{n,m;\textrm{far}}\mid \braket{x^\perp,p^\perp}< 0\}
\end{align*}
so that naturally
\begin{align*}
    W_{n,m;\mathrm{far}}=W_{n,m;\mathrm{out}}\sqcup W_{n,m;\mathrm{in}}
\end{align*}
and let
\begin{align*}
    &\varphi_{n;\textrm{out}}=P_{\delta }(W_{n,m;\textrm{out}})\varphi_n& &\varphi_{n;\textrm{in}}=P_{\delta }(W_{n,m;\textrm{in}})\varphi_n\\
    &\varphi_{n;\textrm{sur}}=P_{\delta }(W_{n,m;\textrm{sur}})\varphi_n
\end{align*}
so that (\ref{eq:b}) holds.
\par
It will be convenient to label the  projections of $W_{n,m;\textrm{in/out}}$ to the transverse coordinates as $W_{n;\textrm{in/out}}^\perp\subset \bbR^{2(d-k)}$ so that $W_{n,m;\textrm{out/in}}=\bbR^{2k}\times W_{n,m;\textrm{out/in}}^\perp$.
  \begin{lemma}\label{omegalem}
 \begin{align}
 &\|(\Omega^--\id)\varphi_{n;\mathrm{out}}\|\xrightarrow{n\rightarrow \infty }0\label{Omega-}\\
 &\|(\Omega^+-\id)\varphi_{n;\mathrm{in}}\|\xrightarrow{n\rightarrow \infty }0\label{Omega+}
 \end{align}
 \end{lemma}
\begin{proof}
We may write
\begin{align}
\begin{split}\label{eq:1}
    \|(\Omega^--\id)\varphi_{n;\textrm{out}}\|&\leq \|\int\limits_0^\infty e^{itH}(H-H_0)e^{-itH_0}\,dt\,\varphi_{n;\textrm{out}}\|\\
    &\leq \int\limits_0^\infty \|Ve^{-itH_0} \varphi_{n;\textrm{out}}\|\,dt\leq M\int\limits_0^\infty \|\chi_{S_{r_0}}e^{-itH_0} \varphi_{n;\textrm{out}}\|\,dt\\
    &\leq  M\int\limits_0^\infty \|\chi_{S_{r_0}}e^{-itH_0} P_{\delta }(W_{n;\textrm{out}})\|_{\textrm{op}}\|\varphi_{n}\|\, dt
\end{split}
\end{align}
since $V$ is supported on $S_{r_0}$ and $M=\|V\|_{\textrm{op}}$. Since we have that $W_{n;\textrm{out}}=\bbR^{2k}\times W_{n;\textrm{out}}^\perp $, we may write $P_\delta(W_{n;\textrm{out}})=\Id\otimes P^\perp_\delta(W_{n;\textrm{out}}^\perp)$. By factoring $\chi_{S_{r_0}}=\Id\otimes B_{r_0}$ and $e^{-itH_0}=e^{-itH_0^\parallel}\otimes e^{-itH_0^\perp}$, we may write
\begin{align*}
    \|\chi_{S_{r_0}}e^{-itH_0} P_{\delta }(W_{n;\textrm{out}})\|_{\textrm{op}}=\|e^{-itH_0^\parallel}\otimes (\chi_{B_{r_0}}e^{-itH_0^\perp}P^\perp_\delta(W_{n;\textrm{out}}^\perp))\|_\textrm{op}=\|\chi_{B_{r_0}}e^{-itH_0^\perp}P^\perp_\delta(W_{n;\textrm{out}}^\perp)\|_\textrm{op}
\end{align*}
because $\|A\otimes B\|_\textrm{op}=\|A\|_\textrm{op}\|B\|_\textrm{op}$ (see \cite{RSVol1}, page 299) and $\|e^{-itH_0^\parallel}\|_\textrm{op}=1 $.
\par
Thus, to proceed we want to show that
\begin{align}\label{MainCookEst}
    \lim_{n\rightarrow\infty}\int\limits_0^\infty\|\chi_{B_{r_0}}e^{-itH_0^\perp} P_{\delta }^\perp(W_{n;\textrm{out}}^\perp)\|_{\textrm{op}}\,dt =0
\end{align}
 from which (\ref{Omega-}) follows in light of (\ref{eq:1}). In what follows, the symbol $C$ refers to such a constant, the exact value of which may change from line to line.\par
From Proposition \ref{PNormBound} we have that
\begin{align*}
    \|\chi_{B_{r_0}}e^{-itH_0^\perp} P_{\delta }^\perp(W_{n;\textrm{out}}^\perp)\|_{\textrm{op}}^2&\leq (2\pi)^{-d} \iint\limits_{W_{n;\textrm{out}}^\perp}\|\chi_{B_{r_0}}e^{-itH_0^\perp}\eta^\perp_{x,p;\delta }\|^2\,dx\,dp
\end{align*}
which we will estimate via the following lemma:
\begin{lemma}[Lemma 2 of Theorem XI.112 in \cite{RSVol3}]\label{NonStationaryRS}
Let $K$ be a compact subset of $\bbR^{\nu}$ and let $\calO$ be an open neighborhood of $K$. Let $\calC(x_0,t)=\{x_0+vt\mid v\in \calO\}$ be the \emph{classically allowed region} for particles starting at $x_0$ with velocities in $\calO$. Then, for any $\ell$ there is a number $\mu$ and a constant $D=D(K,\calO,\ell,d)$ so that:
\begin{align*}
    |e^{-itH_0}u(x)|\leq D(1+d(x,\calC(x_0,t)))^{-\ell}\|(1+|\cdot-x_0|^\mu)u\|
\end{align*}
for all $u$ with $\supp \hat{u}\subset K$ and all $x\in\bbR^\nu$. 
\end{lemma}
In order to apply Lemma \ref{NonStationaryRS} we need the following geometric claims:
\begin{lemma}\label{geoclaim}
For some absolute constant $C$, if $n\geq 8r_0$, and $\delta <\frac{1}{10}m$
\begin{align*}
\|x+t\xi-y\|\geq C(\|x\|+n+t\|p\|)
\end{align*}
for all $(x,p)\in W_{n;\mathrm{out}}^\perp$, $t\geq 0$, $y\in B_{r_0}$, and $\xi\in\calO:=\supp \widehat{\eta^\perp_{x,p;\delta }}+B_{\delta}$.
\end{lemma}
\begin{proof}[proof of claim]
Since $(x,p)\in W_{n;\mathrm{out}}^\perp$  we have that
\begin{align*}
    & \|x\|>n,\,\|p\|>m,\,\textrm{and } \braket{x,p}\geq0
\end{align*}
and we may write $\xi=p+p'$ where $p' \in B_{2\delta}$. It follows that
\begin{align*}
    \frac{\braket{x,\xi}}{\|x\|\|\xi\|}\geq \frac{\braket{x,p'}}{\|x\|\|\xi\|}\geq -\frac{2\delta}{m-2\delta}\geq -\frac{\frac{m}{5}}{m-\frac{m}{5}}=-\frac{1}{4}
\end{align*}
Therefore
\begin{align*}
    \|x+t\xi\|^2&=\|x\|^2+t^2\|\xi\|^2+2t\braket{x,\xi}\geq \|x\|^2+t^2\|\xi\|^2-\frac{t}{2}\|x\|\|\xi\|\\
    &=\frac{3}{8}(\|x\|+t\|\xi\|)^2+\frac{5 }{8}(\|x\|-t\|\xi\|)^2\\ 
    &\geq\frac{3}{8}(\|x\|+t(\|p\|-2\delta))^2\geq \frac{1}{5}(\|x\|+t\|p\|)^2
\end{align*}
Furthermore, since $\|x\|>n$, we may write
\begin{align*}
     \|x+t\xi\|\geq  \frac{1}{5}(\|x\|+n+t\|p\|)
\end{align*} 
Finally, because $\|y\|\leq r_0\leq \frac{1}{8}n$,
\begin{align*}
    \|x+t\xi-y\|\geq \frac{1}{5}(\|x\|+n+\|p\|t)-\|y\|\geq \frac{1}{16}(\|x\|+n+\|p\|t)
\end{align*}
By letting $C=\frac{1}{16}$, we obtain the desired inequality for all $\xi\in \calO$. 
\end{proof}
Let $\calC(x,t)$ be the classically allowed region (see Lemma \ref{NonStationaryRS}) corresponding to $\calO$. For $y$, $\xi$, and $(x,p)$ as above, we have that $y\not\in \calC(x,t)$ so we may apply Lemma \ref{NonStationaryRS} to see that for any $\ell>0$ there is some $\mu>0$ such that
\begin{align*}
    |(e^{-itH_0^\perp}\eta^\perp_{x,p;\delta})(y)|\leq D \frac{\|(1+|\cdot-x|^\mu)\eta^\perp_{x,p;\delta}(\cdot)\|}{d(y,\calC(x,t))^{\ell}}
\end{align*}
uniformly in $(x,p)\in W_{n;\textrm{out}}^\perp$ and $y\in B_{r_0}$. We note that 
\begin{align*}
    \|(1+|\cdot-x|^\mu)\eta^\perp_{x,p;\delta}(\cdot)\|\leq \|\eta_{x,p;\delta}^\perp\|+\left(\,\int\limits_{\bbR^d}\|y-x\|^{2\mu}|\eta^\perp_\delta(y-x)|^2\,dy\right)^\frac{1}{2}
\end{align*}
where the latter expression is independent of $x$ and $p$ (but depends on $\delta$) and is finite since
${\eta^\perp_\delta \in \calS(\bbR^{d-k})}$.
Therefore, for $(x,p) \in W_{n;\textrm{out}}^\perp$
\begin{align*}
    &\|\chi_{B_{r_0}}e^{-itH_0}\eta^\perp_{x,p;\delta }\|^2\leq C\int\limits_{B_{r_0}}(\|x\|+n+t\|p\|)^{-\ell}\,dy\leq C(\|x\|+n+t\|p\|)^{-\ell}
\end{align*}
Using the above, for any $\ell$ large enough relative to $d-k$, we may write
\begin{align*}
    &\iint\limits_{W_{n;\textrm{out}}^\perp}\|\chi_{B_{r_0}}e^{-itH_0^\perp}\eta^\perp_{x,p;\delta }\|^2\,dx\,dp\leq C\iint\limits_{W_{n;\textrm{out}}^\perp}(\|x\|+n+t\|p\|)^{-\ell}\,dx\,dp\\
    &\leq C\int\limits_{B_{m}^c}\int\limits_{B_n^c} (\|x\|+n+t\|p\|)^{-\ell}\,dx\,dp\\
    &\leq C\int\limits_{B_m^c}\int\limits_{n}^\infty  (r+n+t\|p\|)^{-\ell}r^{d-k-1}\,dr\,dp\\
    &\leq C\int\limits_m^\infty(n+t\rho)^{-\ell+d-k}\rho^{d-k-1}\,d\rho\leq Ct^{-1}(n+tm)^{-\ell+2(d-k)}
\end{align*}
Thus, we may conclude that
\begin{align}\label{MainCookInEta}
     \|\chi_{B_{r_0}}e^{-itH_0^\perp} P_{\delta }^\perp(W_{n;\textrm{out}}^\perp)\|_{\textrm{op}}^2&\leq Ct^{-1}(n+tm)^{-\ell+2(d-k)}
\end{align}
\par
To see (\ref{MainCookEst}), we first note that for all $t$ and $n$
\begin{align*}
    \|\chi_{B_{r_0}}e^{-itH_0^\perp}P_\delta^\perp(W^\perp_{n;\textrm{out}})\|_\textrm{op}\leq 1
\end{align*}
so that by combining the two bounds and choosing $\ell$ sufficiently large we may write
\begin{align*}
    &\int\limits_0^\infty\|\chi_{B_{r_0}}e^{-itH_0^\perp}P_\delta^\perp(W^\perp_{n;\textrm{out}})\|_\textrm{op}\,dt\leq\int\limits_0^\frac{1}{n}1 \,dt+C\int\limits_{\frac{1}{n}}^\infty t^{-\frac{1}{2}}(n+tm)^{-\ell}\, dt \\
    &\leq \frac{1}{n}+C\sqrt{n}\int\limits_{\frac{1}{n}}^\infty (n+tm)^{-\ell} \,dt=\frac{1}{n}+C\frac{\sqrt{n}}{m} (n+\frac{m}{n})^{-\ell+1}
\end{align*}
which proves (\ref{Omega-}). The limit (\ref{Omega+}) may be deduced from exactly the same argument by first writing
\begin{align*}
    &\|(\Omega^+-\id)\varphi_{n;\textrm{in}}\|\leq M\int\limits_{-\infty}^0 \|\chi_{S_{r_0}}e^{-iH_0t} \varphi_{n;\textrm{in}}\|\,dt
\end{align*}
and noting that for $t\leq 0$, $e^{-itH_0}\varphi_{n;\textrm{in}}$ behaves like $e^{-itH_0}\varphi_{n;\textrm{out}}$ for $t\geq 0$ because $ W_{n;\textrm{out}}^\perp$ and $W_{n;\textrm{in}}^\perp$ are related by $(x,p)\mapsto (x,-p)$.
\end{proof} 
\begin{lemma}\label{Inlemma}
If we assume that $\varphi_n=e^{-it_nH}\varphi$ for some sequence of positive times $\{t_n\}_{n=0}^\infty $, then
\begin{align*}
   \|\varphi_{n;\mathrm{in}}\|\xrightarrow{n\rightarrow \infty }0 
\end{align*}
\end{lemma}
\begin{proof}
This proof is based on an argument of Enss recorded in \cite{SimonEnns}.
We can write
\begin{align*}
    &\|\varphi_{n;\textrm{in}}\|=\|P_{\delta }(W_{n;\textrm{in}})e^{-it_nH}\varphi\|\\
    &\leq\|P_{\delta }(W_{n;\textrm{in}})(e^{-it_nH}-e^{-it_nH_0})\varphi\|+\|P_{\delta }(W_{n;\textrm{in}})e^{-it_nH_0}\varphi\| 
\end{align*}
so it suffices to prove that
\begin{align}
    \|P_{\delta }(W_{n;\textrm{in}})(e^{-it_nH}-e^{-it_nH_0})\|_{\textrm{op}}\xrightarrow{n\rightarrow\infty}0\label{eq:2}
\end{align}
and
\begin{align}
   \slim\limits_{n\rightarrow\infty} P_{\delta }(W_{n;\textrm{in}})e^{-it_nH_0}=0 \label{eq:3}
\end{align}
To prove (\ref{eq:2}), we write
\begin{align*}
    &\|P_{\delta }(W_{n;\textrm{in}})(e^{-it_nH}-e^{-it_nH_0})\|_\textrm{op}= \|(e^{it_nH}-e^{it_nH_0})P_{\delta }(W_{n;\textrm{in}})\|_\textrm{op}\\
    &=\|(\id-e^{-it_nH}e^{it_nH_0})P_{\delta }(W_{n;\textrm{in}})\|_\textrm{op}\leq \int\limits_0^{t_n}\|Ve^{i\tau H_0}P_{\delta }(W_{n;\textrm{in}})\|_\textrm{op}\,d\tau\\
    &\leq M\int\limits_0^{t_n}\|\chi_{S_{r_0}}e^{i\tau H_0}P_{\delta }(W_{n;\textrm{in}})\|_\textrm{op}\,d\tau\leq M\int\limits_{0}^{\infty}\|\chi_{S_{r_0}}e^{i\tau H_0}P_{\delta }(W_{n;\textrm{in}})\|_\textrm{op}\,d\tau
\end{align*}
By using (\ref{MainCookInEta}) and the symmetry between $ W_{n;\textrm{out}}^\perp$ and $W_{n;\textrm{in}}^\perp$ when mapping  $(x,p)\mapsto (x,-p)$ we see that for any $\ell>0$
\begin{align*}
    \|\chi_{S_{r_0}}e^{i\tau H_0}P_{\delta }(W_{n;\textrm{in}})\|_\textrm{op}\leq C\tau^\frac{1}{2} (n+m\tau)^{-\ell}
\end{align*}
as long as $\tau>0$, so we conclude, similarly to the above, that
\begin{align*}
    \int\limits_0^{\infty}\|\chi_{S_{r_0}}e^{i\tau H_0}P_{\delta }(W_{n;\textrm{in}})\|_\textrm{op}\,d\tau\xrightarrow{n\rightarrow \infty}0
\end{align*}
thus establishing (\ref{eq:2}).\par
For (\ref{eq:3}), we fix $\psi \in \calH$ compactly supported and choose $R$ so that $\supp \psi \subset S_R$. Then
\begin{align*}
    &\|P_{\delta }(W_{n;\textrm{in}})e^{-iH_0t_n}\psi \|=\|P_{\delta }(W_{n;\textrm{in}})e^{-iH_0t_n}\chi_{S_R}\psi \|\\
    &\leq \|\chi_{S_R} e^{iH_0t_n}P_{\delta }(W_{n;\textrm{in}})\|_\textrm{op} \|\psi\|\xrightarrow{n\rightarrow \infty }0
\end{align*}
because the computation of the above operator norm applies just as well to $S_R$ for $R>0$ arbitrary instead of $S_{r_0}$.\par
Density establishes (\ref{eq:3}), which concludes the proof of the lemma.
\end{proof}
These lemmas establish Theorem \ref{thm:decomposition} in full. 
\end{proof}
\section{Proof of part \ref{Complete} of Theorem \ref{thms}: Asymptotic completeness}\label{SumProof}
Recall that 
\begin{align*}
    \tilde{\calH}_{\textrm{sur}}(H)=\bigcap_{m>0}\{\psi\in\calH\mid
     \calN_H^{m}(\psi) =0\}
\end{align*}
where
\begin{align*}
    \calN_H^{m}(\psi)=\limsup\limits_{\delta \rightarrow 0 }\limsup\limits_{n\rightarrow\infty} \sup\limits_{t\geq 0 }\|P_\delta(W_{n,m;\textrm{far}})e^{-itH} \psi\|
\end{align*}
The proof is accomplished in three steps: the first is to prove that $\Ran\Omega^-$ and $\tilde{\calH}_\textrm{sur}$ span all of $\calH$, the second is to show that their intersection is $0$, and the third is to prove that $\tilde{\calH}_\textrm{sur}=\calH_\textrm{sur}$. 
\subsection{Step 1: The Span of $\Ran\Omega^-$ and $\tilde{\calH}_\mathrm{sur}$}\label{Spans}
The above decomposition theorem (Theorem \ref{thm:decomposition}), establishes the first step towards the proof of part \ref{Complete} of Theorem \ref{thms}:
\begin{lemma}\label{sum}
\begin{align*}
\calH=\tilde{\calH}_\mathrm{sur}+\Ran(\Omega^-)
\end{align*}
\end{lemma}
\begin{proof}
 Let $\varphi\in (\Ran (\Omega^-))^\perp$. 
Fix $m>0$ and for each $n$ choose $t_n\geq 0$ such that
\begin{align*}
    \|P_{\delta }(W_{n,m;\textrm{far}})e^{-it_nH}\varphi\|>\frac{1}{2}\sup_{t\geq 0} \|P_{\delta }(W_{n,m;\textrm{far}})e^{-itH}\varphi\|
\end{align*}
Let $\varphi_n=e^{-it_nH}\varphi$. By Theorem \ref{thm:decomposition}, there is a $\delta_0(m)$ such that for all $\delta \in (0,\delta_0)$, there is a decomposition depending on $m$ and $\delta$
\begin{align*}
     \varphi_n=\varphi_{n;\textrm{out}}+\varphi_{n;\textrm{in}}+\varphi_{n;\textrm{sur}}
\end{align*}
obeying the properties in the theorem.\par
Now, since $\varphi\perp \textrm{Ran}(\Omega^-)$, from property (\ref{eq:a}) in Theorem \ref{thm:decomposition} we get that
\begin{align*}
    \lim\limits_{n\rightarrow \infty }|\braket{\varphi_n, \varphi_{n;\textrm{out}}}|=\lim\limits_{n\rightarrow \infty }|\braket{\varphi_n, \Omega^-\varphi_{n;\textrm{out}}}|=0
\end{align*}
where we have also used that the propagator leaves $\textrm{Ran}(\Omega^-)$ invariant. Furthermore, ${\|\varphi_{n;\textrm{in}}\|\xrightarrow{n\rightarrow\infty} 0}$ so that from property (\ref{eq:b}) in Theorem \ref{thm:decomposition} we get that
\begin{align*}
    \braket{P_{\delta } (W_{n;\textrm{far}})\varphi_n,\varphi_n}=\braket{\varphi_{n;\textrm{in}}+\varphi_{n;\textrm{out}},\varphi_n}\xrightarrow{n\rightarrow\infty}0
\end{align*}
But this implies that
\begin{align*}
\|P_{\delta }(W_{n;\textrm{far}})\varphi_n\|\xrightarrow{n\rightarrow\infty}0
\end{align*}
as
\begin{align*}
    \|P_{\delta }(W_{n;\textrm{far}})\varphi_n\|^2=\braket{P_{\delta }^2(W_{n;\textrm{far}})\varphi_n,\varphi_n}\leq \braket{P_{\delta }(W_{n;\textrm{far}})\varphi_n,\varphi_n}
\end{align*}
since $P_\delta^2(E)\leq P_\delta(E)$ as in the proof of Proposition \ref{BoundedByIndicator}.
So we get that 
\begin{align*}
    \frac{1}{2}\sup_{t\geq 0} \|P_{\delta }(W_{n,m;\textrm{far}})e^{-itH}\varphi\|\leq \|P_{\delta }(W_{n,m;\textrm{far}}\varphi_n\|\xrightarrow{n\rightarrow\infty }0 
\end{align*}
which implies that
\begin{align*}
    \lim\limits_{n\rightarrow \infty} \sup\limits_{t\geq 0}\|P_{\delta }(W_{n;\textrm{far}})e^{-itH}\varphi\|=0
\end{align*}
for all $\delta<\delta_0(m)$. In other words, since the choice of $m$ was arbitrary, we have shown that $\varphi \in \tilde{\calH}_\textrm{sur}$. \par
\end{proof}
\subsection{Step 2: The Intersection of $\Ran\Omega^-$ and $\tilde{\calH}_\mathrm{sur}$}\label{Intersection}
\begin{lemma}\label{SemiComplete}
\begin{align*}
    \calH=\mathrm{Ran}(\Omega^-)\oplus \tilde{\calH}_{\mathrm{sur}}
\end{align*}
\end{lemma}
\begin{proof}
By Lemma \ref{sum}, it suffices to show that $\textrm{Ran}(\Omega^-)\cap \tilde{\calH}_{\textrm{sur}}=\{0\}$. For this, we will define the following auxiliary family of seminorms
\begin{align*}
    \tilde{\calN}_H^{m}(\psi)=\limsup\limits_{\delta \rightarrow 0 }\limsup\limits_{n\rightarrow\infty} \limsup\limits_{t\rightarrow \infty }\|P_\delta(W_{n,m;\textrm{far}})e^{-itH} \psi\|
\end{align*}
by replacing $\sup\limits_{t\geq0}$ in the definition of $\calN_H^{m}$ with $\limsup\limits_{t\rightarrow\infty}$ so that we may define
\begin{align*}
    \calH_{\mathrm{sur}}^{\limsup}(H):=\bigcap_{m>0}\{\psi\in\calH\mid
     \tilde{\calN}_H^{m}(\psi) =0\}
\end{align*} accordingly. Clearly, $\tilde{\calH}_\textrm{sur}\subset \calH_{\textrm{sur}}^{\limsup}(H)$ and will prove the stronger claim that
\begin{align*}
    \Ran(\Omega^-)\cap \calH_{\textrm{sur}}^{\limsup}(H)=\{0\}
\end{align*}
To this end, we first prove:
\begin{claim}\label{IsoIntersection}
If $\psi\in \Ran(\Omega^-)\cap \calH_{\mathrm{sur}}^{\limsup}(H)$ then there exists $\varphi\in\calH_{\mathrm{sur}}^{\limsup}(H_0)$ such that $\Omega^-(\varphi)=\psi$.
\end{claim}
\begin{proof}
Let $\psi \in \textrm{Ran}(\Omega^-)\cap \calH_{\textrm{sur}}^{\limsup}(H)$. Since $\psi \in \textrm{Ran}(\Omega^-)$ there is some $\varphi\in \calH $ such that:
\begin{align*}
    \psi=\Omega^-\varphi
\end{align*}
or equivalently
\begin{align*}
    \lim\limits_{t\rightarrow\infty}\|e^{-itH_0}\varphi-e^{-itH}\psi\|=0
\end{align*}
For any fixed $m>0, \delta>0,n>0$, this implies that
\begin{align*}
    &\limsup\limits_{t\rightarrow\infty}\|P_\delta(W_{n,m;\textrm{far}}) e^{-itH}\psi - P_\delta(W_{n,m;\textrm{far}})e^{-itH_0}\varphi\|\\
    &\leq \limsup\limits_{t\rightarrow\infty}\|P_\delta(W_{n,m;\textrm{far}}) \|_{\textrm{op}}\|e^{-itH}\psi-e^{-itH_0}\varphi\|=0
\end{align*}
so that 
\begin{align*}
     &\limsup\limits_{t\rightarrow \infty } \|P_\delta(W_{n,m;\textrm{far}})e^{-itH_0} \varphi\| \\
     &\leq\limsup\limits_{t\rightarrow \infty } \left[\|P_\delta(W_{n,m;\textrm{far}})(e^{-itH_0} \varphi-e^{-itH} \psi)\|+\|P_\delta(W_{n,m;\textrm{far}})e^{-itH} \psi\| \right]\\
     &\leq \limsup\limits_{t\rightarrow \infty } \|P_\delta(W_{n,m;\textrm{far}})e^{-itH} \psi\|
\end{align*}
Therefore
\begin{align*}
    &\limsup\limits_{\delta\rightarrow0}\limsup_{n\rightarrow\infty}\limsup_{t\rightarrow\infty}\|P_\delta(W_{n,m;\textrm{far}})e^{-itH_0} \varphi\|\\
    &\leq \limsup\limits_{\delta\rightarrow0}\limsup_{n\rightarrow\infty} \limsup_{t\rightarrow \infty}\|P_\delta(W_{n,m;\textrm{far}})e^{-itH} \psi\|=0
\end{align*}
since $\psi \in \calH_{\textrm{sur}}^{\limsup}(H)$ so that because $m$ was arbitrary, we see that $\varphi \in \calH_{\textrm{sur}}^{\limsup}(H_0)$.
\end{proof}
Thus, it suffices to show that $\calH_{\textrm{sur}}^{\limsup}(H_0)=\{0\}$:
\begin{claim}\label{FreeInter}
\begin{align*}
    \calH_{\mathrm{sur}}^{\limsup}(H_0)=\{0\}
\end{align*}
\end{claim}

\begin{proof}
Recall the following definition:
\begin{align*}
    \calD_{\alpha}=\textrm{Span}(\{\psi_i^\parallel\otimes \psi_i^\perp \mid \psi^\parallel_i\in L^2(\bbR^k), \psi ^\perp_i \in \calS(\bbR^{d-k}), \supp \widehat{\psi^\perp_i} \Subset B_\alpha^c\}
\end{align*}
for some $\alpha>0$. \\
 We will show that $\calD_{\alpha}\cap \calH_{\textrm{sur}}^{\limsup}(H_0)=\{0\}$ from which the claim follows by the density of $\bigcup\limits_{\alpha>0} \calD_{\alpha}$ in $\calH$.\par
For this, fix $\varphi\in\calD_{\alpha}\cap \calH_{\textrm{sur}}^{\limsup}(H_0)$ and
choose $m<\alpha$. For all $\delta<\alpha-m $ sufficiently small
\begin{align*}
    P_\delta(W_{n,m;\textrm{sur}})\varphi=P_\delta(\bbR^{2k}\times B_n\times \bbR^{d-k})\varphi
\end{align*}
 as $ P_\delta(W_{n,m;\textrm{sur}})=P_\delta(\bbR^{2k}\times B_n\times \bbR^{d-k})+P_\delta(\bbR^{2k}\times B_n^c\times B_m)$ and $P_\delta(\bbR^{2k}\times B_n^c\times B_m)\varphi=0$ by Proposition \ref{psupport}. Furthermore, this equality holds for $e^{-itH_0}\varphi$ for all $t$ since the free propagator does not change a function's Fourier support.\par
 By using Proposition \ref{SpaceLoc} we see that, (since $S_n=\bbR^k \times B_n$)
\begin{align*}
    \|P_\delta(W_{n,m;\textrm{sur}})e^{-itH_0}\varphi\|&=\|P_\delta(\bbR^{2k}\times B_n\times \bbR^{d-k})e^{-itH_0}\varphi\|= \|(|\eta_\delta|^2*\chi_{S_{n}})e^{-itH_0}\varphi\|\\
    &\leq\|(|\eta_\delta|^2*\chi_{S_{n}})\chi_{S_{2n}}e^{-itH_0}\varphi\|+\|(|\eta_\delta|^2*\chi_{S_{n}})\chi_{S_{2n}^c}e^{-itH_0}\varphi\|\\
    &\leq \|\chi_{S_{2n}}e^{-itH_0}\varphi\|+\|\varphi\|\|(|\eta_\delta|^2*\chi_{S_{n}})\chi_{S_{2n}^c}\|_\infty\\
    &= \|\chi_{S_{2n}}e^{-itH_0}\varphi\|+\|\varphi\|\|(|\eta_\delta^\perp|^2*\chi_{B_{n}})\chi_{B_{2n}^c}\|_\infty
\end{align*}
Note that for $x\in B_{2n}^c$ and $\ell>0$ large enough
\begin{align}\label{SmoothIndToInd}
\begin{split}
     &(|\eta_\delta^\perp|^2*\chi_{B_{n}})(x)=\int\limits_{B_{n}}|\eta_\delta^\perp(y-x)|^2\,dy\leq C\int\limits_{B_{n}}\|x-y\|^{-\ell}\,dy\leq C\int\limits_{B_{n}}(2n-\|y\|)^{-\ell}\,dy\\
   &\leq C\int\limits_{0}^n(2n-r)^{-\ell}r^{d-k-1}\,dr\leq Cn^{-\ell+d-k}
\end{split}
\end{align}
so we can write, for any $\ell>0$
\begin{align*}
    \|P_\delta(W_{n,m;\textrm{sur}})e^{-itH_0}\varphi\|\leq \|\chi_{S_{n}}e^{-itH_0}\varphi\|+C\|\varphi\|n^{-\ell }
\end{align*}
Thus, we conclude that
 \begin{align*}
     \lim\limits_{n\rightarrow \infty} \lim\limits_{t\rightarrow\infty} \|P_\delta(W_{n,m;\textrm{sur}})e^{-itH_0}\varphi\|\leq \lim\limits_{n\rightarrow \infty} \lim\limits_{t\rightarrow\infty} \|\chi_{S_{2n}}e^{-itH_0}\varphi\|
 \end{align*}
 To estimate the right hand side, we note that $x\in S_n$ implies that for $t>\frac{n}{2\alpha}$ we have
 \begin{align*}
     \|\frac{x^\perp}{t}\|<\frac{n}{t}<2\alpha
 \end{align*}
  Therefore, we may proceed as in the proof of (\ref{ExistEq}) in the proof of part \ref{existenceTheorem} of Theorem \ref{thms} to see that for all $\ell>0$ we have
\begin{align*}
    \|\chi_{S_n}e^{-itH_0}\varphi\|^2\leq C \int\limits_{B_n}(1+\|x^\perp\|+t)^{-\ell}\,dx^\perp
\end{align*}
and therefore
\begin{align*}
    \lim\limits_{t\rightarrow\infty}\|\chi_{S_n}e^{-itH_0}\varphi\|=0
\end{align*}
In summary, we have shown that we may find $m_0(\alpha)$  such that for some $\delta_0$ if $\delta<\delta_0$ then
\begin{align*}
    \lim\limits_{n\rightarrow\infty}\lim\limits_{t\rightarrow\infty} \|P_\delta(W_{n,m_0;\textrm{sur}})e^{-itH_0}\varphi\|=0
\end{align*}
Now, because $\varphi\in \calH_{\textrm{sur}}^{\limsup}(H_0)$, we may find $\delta<\delta_0$ so that
\begin{align*}
    \lim\limits_{n\rightarrow\infty}\limsup\limits_{t\rightarrow\infty}\|P_\delta(W_{n,m_0;\textrm{far}})e^{-itH_0}\varphi\|<\varepsilon
\end{align*}
for any $\varepsilon>0$. It follows that
\begin{align*}
    \|\varphi\|=\lim\limits_{n\rightarrow \infty} \limsup\limits_{t\rightarrow \infty }\|P_{\delta}(W_{n,m_0;\textrm{sur}})e^{-itH_0}\varphi+P_{ \delta}(W_{n,m_0;\textrm{far}})e^{-itH_0}\varphi\|< \varepsilon
\end{align*}
and since $\varepsilon$ was arbitrary, we see that $\varphi=0$.
\end{proof}
These claims complete the proof of Lemma \ref{SemiComplete}.
\end{proof}
It may be of interest to note that we have in fact proven that it is equivalent to define $\tilde{\calH}_\textrm{sur}$ with a $\limsup$ in time instead of a $\sup$. In other words:
\begin{corollary}
We have that
\begin{align*}
    \calH_{\mathrm{sur}}^{\limsup}(H)=\tilde{\calH}_{\mathrm{sur}}
\end{align*}
\end{corollary}
\begin{proof}
We have shown that
\begin{align*}
 \textrm{Ran}(\Omega^-)\cap \calH_{\textrm{sur}}^{\limsup}(H)=\{0\}
\end{align*}
and since $\calH=\Ran(\Omega^-)\oplus \tilde{\calH}_\textrm{sur}$ and $\tilde{\calH}_{\mathrm{sur}}\subset \calH_{\mathrm{sur}}^{\limsup}(H)$, we have in addition that:
\begin{align*}
    \calH=\Ran(\Omega^-)+\calH_{\textrm{sur}}^{\limsup}(H)
\end{align*}
And the desired equality follows immediately.
\end{proof}
\subsection{Step 3: $\tilde{\calH}_{\mathrm{sur}}=\calH_{\mathrm{sur}}$}\label{SurSpaceChar}
\begin{proof}[proof of part \ref{Complete} of Theorem \ref{thms}]
Recall the definition of the surface subspace:
\begin{align*}
    \calH_{\textrm{sur}}=\{\psi \in \calH \mid \forall v>0, \lim\limits_{t\rightarrow \infty}\|\chi_{S_{vt}}e^{-itH} \psi\|= \|\psi\|\}
\end{align*}
The proof of the desired equality will lean on a non-stationary phase argument:
\begin{lemma}\label{VertClaim}
Fix $v>0$. For any $m<\frac{v}{16}$ and $\delta<\frac{m}{2}$ and for any $\psi \in \calH$ we have
\begin{align*}
   \lim\limits_{t\rightarrow\infty}\|P_\delta(\bbR^{2k}\times B_{vt}^c\times B_m)e^{-itH} \psi \|=0 
\end{align*}
\end{lemma}
\begin{proof}
Denote $A_{vt,m}=\bbR^{2k}\times B_{vt}^c\times B_m$,  then we can write
\begin{align*}
    &\|P_\delta(A_{vt,m})e^{-itH} \psi\|\leq\|P_\delta(A_{vt,m})(e^{-itH}-e^{-itH_0})\psi\|+\|P_\delta(A_{vt,m})e^{-itH_0}\psi\|
\end{align*}
so it suffices to prove that
\begin{align}
    \|P_\delta(A_{vt,m})(e^{-itH}-e^{-itH_0})\|_{\textrm{op}}\xrightarrow{t\rightarrow\infty}0\label{eqvert:2}
\end{align}
and
\begin{align}
   \slim\limits_{t\rightarrow\infty} P_\delta(A_{vt,m})e^{-itH_0}=0 \label{eqvert:3}
\end{align}
As before, both claims will follow from an estimate on the free propagation $e^{-itH_0}$.
\begin{claim}
With all parameters as above, for any $R>0$ and $\ell>0$, there exists $C>0$, independent of $t,m$ and $v$ such that
\begin{align}\label{mainestvert}
    \|\chi_{S_R} e^{-iH_0t}P_\delta(\bbR^{2k}\times B_{v|t|}^c\times B_m)\|_\textrm{op}\leq C(v|t|)^{-\ell}
\end{align}
for all $|t|>\frac{8R}{v}$.
\end{claim}
\begin{proof}
We will first prove the claim for $t>\frac{8R}{v}$. Note that
\begin{align*}
    \chi_{S_R}e^{-itH_0}P_\delta(\bbR^{2k}\times B_{vt}^c\times B_m)=e^{-iH_0^\parallel t}\otimes (\chi_{B_R}e^{-iH_0^\perp t}P_{\delta}(B_{vt}^c\times B_m))
\end{align*}
so that
\begin{align*}
    \|\chi_{S_R} e^{-iH_0t}P_\delta(\bbR^{2k}\times B_{vt}^c\times B_m)\|_\textrm{op}=\|\chi_{B_R}e^{-iH_0^\perp t}P_{\delta}(B_{vt}^c\times B_m)\|_\textrm{op}
\end{align*}
From Proposition \ref{PNormBound} we have that
\begin{align}\label{HSurPNormBound}
    \|\chi_{B_{R}}e^{-itH_0^\perp}P_\delta(B_{vt}^c\times B_m)\|_{\textrm{op}}^2&\leq (2\pi)^{-d} \iint\limits _{B_{vt}^c\times B_m}\|\chi_{B_{R}}e^{-itH_0^\perp}\eta^\perp_{x,p;\delta }\|^2\,dx\,dp
\end{align}
Let $ x\in B_{vt}^c,p\in B_m,y \in B_R$, and $\xi\in \calO:=\supp \hat{\eta}_{x,p;\delta }+B_{\delta}$. For $\xi\in \calO$, we may write $\xi =p+p'$, where $p'\in B_{2\delta}$, which implies
\begin{align*}
    \|\xi\|&\leq m+2\delta \leq\frac{v}{8}
\end{align*}
It follows that
\begin{align*}
    \|x+\xi t-y\|&\geq
    \|x\|-\|y\|-t\|\xi\|\geq \|x\|-R-\frac{v}{8}t\geq \frac{1}{16}(\|x\|+vt)
\end{align*}
where we have used that $\|x\|>vt>8R>\|y\|$. \par 
As in the proof of Lemma \ref{omegalem}, we may apply Lemma \ref{NonStationaryRS} to see that
for any $\ell>0$ there is some $C$ such that
\begin{align*}
    |e^{-itH_0^\perp}\eta^\perp_{x,p;\delta}(y)|\leq C(\|x\|+vt)^{-\ell}
\end{align*}
uniformly in $x,p$ and $y$ as above and $t\geq0$. Since $R$ is fixed
\begin{align*}
\|\chi_{B_{R}}e^{-itH_0^\perp}\eta^\perp_{x,p;\delta }\|^2\leq C(\|x\|+vt)^{-\ell} 
\end{align*}
uniformly in $x$ and $p$, therefore we may integrate (\ref{HSurPNormBound})
 to find that
 \begin{align*}
    \|\chi_{B_{R}}e^{-itH_0^\perp }P_\delta(B_{vt}^c\times B_m)\|_{\textrm{op}}^2&\leq C(vt)^{-\ell+d-k}
\end{align*}
Furthermore, since $B^c_{vt}\times B_m$ is invariant under $(x,p)\mapsto (x,-p)$, the claim holds for $t<-\frac{8R}{v}$ as well.
\end{proof}
 The limit (\ref{eqvert:2}), as in Lemma \ref{Inlemma}, follows from the bound
\begin{align*}
    \|P_\delta(A_{vt,m})(e^{-itH}-e^{-itH_0})\|\leq M\int\limits_0^t \|\chi_{ S_{r_0}}e^{i\tau H_0}P_\delta(A_{vt,m})\,d\tau\|_\textrm{op}
\end{align*}
and the above claim. The limit (\ref{eqvert:3}) may be established by noting that for $\psi$ such that $\supp \psi\subset S_R$ we can write
\begin{align*}
    &\|P_{\delta }(A_{vt,m})e^{-iH_0t}\psi \|=\|P_{\delta }(A_{vt,m})e^{-iH_0t_n}\chi_{S_R}\psi \|\leq \|\chi_{S_R} e^{iH_0t_n}P_{\delta }(A_{vt,m})\|_\textrm{op} \|\psi\|\xrightarrow{t\rightarrow \infty }0
\end{align*}
by the above. Since such $\psi$ are dense, the lemma is proven.
\end{proof} 
\begin{proposition}
 \begin{align*}
 \calH_{\mathrm{sur}}=\tilde{\calH}_{\mathrm{sur}}
 \end{align*}
\end{proposition}
 
\begin{proof}
We will start by showing that $\tilde{\calH}_{\mathrm{sur}}\subset \calH_{\textrm{sur}}$.\par
For this, choose $\psi\in\tilde{H}_\mathrm{sur}$ and fix $v>0$ for which we must show
\begin{align*}
    \lim_{t\rightarrow\infty}\|\chi_{S_{vt}}e^{-itH}\psi\|=\|\psi\|
\end{align*}
To see this, fix $\varepsilon>0$ and choose $m<\frac{v}{16}$. Since $\psi \in \tilde{\calH}_{\mathrm{sur}}$, we know that for this $m$ there is some $\delta$ such that 
\begin{align*}
    \limsup\limits_{n\rightarrow\infty} \sup\limits_{t\geq 0 }\|P_\delta(W_{n,m;\textrm{far}})e^{-itH} \psi\|<\varepsilon
\end{align*}
and thus for $T_0$ large enough, if $t>T_0$ then
\begin{align*}
    \sup\limits_{\tau\geq 0 }\|P_\delta(W_{vt,m;\textrm{far}})e^{-i\tau H} \psi\|<\varepsilon
\end{align*}
Recalling that $W_{n,m;\textrm{sur}}=\bbR^{2k}\times B_{n}\times \bbR^{d-k} \cup \bbR^{2k}\times B_{n}^c\times B_m$, we now write
\begin{align*}
    &\|\psi\|= \|(P_\delta(W_{vt,m;\textrm{far}})+P_\delta(\bbR^{2k}\times B_{vt}\times \bbR^{d-k})+P_\delta(\bbR^{2k}\times B_{vt}^c\times B_m))e^{-itH} \psi\|\\
    &\leq \sup\limits_{\tau\geq 0}\|P_\delta(W_{vt,m;\textrm{far}})e^{-i\tau H} \psi\|+\|P_\delta(\bbR^{2k}\times B_{vt}\times \bbR^{d-k})e^{-itH} \psi\|+\|P_\delta(\bbR^{2k}\times B_{vt}^c\times B_m)e^{-itH} \psi\|\\
    &\leq\varepsilon+\|P_\delta(\bbR^{2k}\times B_{vt}\times \bbR^{d-k})e^{-itH} \psi\|+\|P_\delta(\bbR^{2k}\times B_{vt}^c\times B_m)e^{-itH}\psi\|
\end{align*}
for all $t>T_0$.
By Proposition \ref{SpaceLoc}, we may estimate the second term
\begin{align*}
    &\|P_\delta(\bbR^{2k}\times B_{vt}\times \bbR^{d-k})e^{-itH} \psi\|=\|(\eta_\delta* \chi_{S_{vt}})e^{-itH} \psi\|\\
    &\leq \|(\eta_\delta* \chi_{S_{vt}})\chi_{S_{2vt}}e^{-itH} \psi\|+\|(\eta_\delta* \chi_{S_{vt}})\chi_{S_{2vt,m}^c}e^{-itH} \psi\|\\
    &\leq \|\chi_{S_{2vt}}e^{-itH}\psi\|+\|(\eta_\delta* \chi_{S_{vt}})\chi_{S_{2vt}^c}\|_{\textrm{op}}\|\psi\|
\end{align*}
Using (\ref{SmoothIndToInd}), we see that for some $\ell>0$
\begin{align*}
    \|\psi\|\leq \varepsilon+\|\chi_{S_{2vt}}e^{-itH}\psi\|+C(vt)^{-\ell}+\|P_\delta(\bbR^{2k}\times B_{vt}^c\times B_m)e^{-itH}\psi\|
\end{align*}
By Lemma \ref{VertClaim}, taking the limit as $t\rightarrow\infty$ implies that
\begin{align*}
\|\psi\|&\leq\varepsilon+\lim_{t\rightarrow\infty}\|\chi_{S_{2vt}}e^{-itH}\psi\|
\end{align*}
Since $\varepsilon$  was arbitrary and $\|\chi_{S_{2vt}}e^{-itH}\psi\|\leq\|\psi\|$ we may conclude that
\begin{align*}
    \lim\limits_{t\rightarrow\infty } \|\chi_{S_{2vt}}e^{-itH}\psi\|=\|\psi\|
\end{align*}\par
To complete the proof, we will show that $ \calH_{\textrm{sur}}\perp \textrm{Ran}(\Omega^-)$, which implies that $\calH_\textrm{sur}\subset \tilde{\calH}_\textrm{sur}$. In fact, we will show that $\calH_{\textrm{sur}}\perp \Omega^-(\calD_{\alpha})$ for any $\alpha>0$ and conclude by density.\par
Let $\psi \in\calH_{\textrm{sur}}, \varphi\in \Omega^-(\calD_{\alpha})$. Note that the definition of $\calH_{\textrm{sur}}$ implies 
\begin{align}\label{lim:Hsur2}
   \lim\limits_{t\rightarrow\infty}\|\chi_{S_{vt}^c}e^{-itH}\psi\|=0
\end{align}
 For any $v>0$ and any $t>0$, by writing
\begin{align*}
    |\braket{\psi,\varphi}|&\leq |\braket{\chi_{S_{vt}^c}e^{-itH}\psi,e^{-itH}\varphi}|+|\braket{e^{-itH}\psi,\chi_{S_{vt}}e^{-itH}\varphi}|\\
    &\leq \|\chi_{S_{vt}^c}e^{-itH}\psi\|\|\varphi \|+\|\psi\|\|\chi_{S_{vt}}e^{-itH}\varphi\|
\end{align*}
and then taking a $\lim$ as $t\rightarrow\infty$ we see that
\begin{align}\label{ineq:Hsur2}
     |\braket{\psi,\varphi}|&\leq \lim\limits_{t\rightarrow\infty}\left[\|\chi_{S_{vt}^c}e^{-itH}\psi\|\|\varphi \|+\|\psi\|\|\chi_{S_{vt}}e^{-itH}\varphi\|\right]
\end{align}
Now, choose $v<2\alpha$, since $\varphi\in \Omega^-(\calD_{\alpha})$, there is some $\tilde{\varphi}\in \calD_{\alpha} $ such that
\begin{align*}
    \|e^{-itH}\varphi- e^{-itH_0}\tilde{\varphi}\|\xrightarrow{t\rightarrow \infty }0
\end{align*}
 Next, because
 \begin{align*}
    &\tilde{\varphi}=\sum_{i=1}^n \tilde{\varphi}^\parallel_i\otimes \tilde{\varphi}^\perp_i, \,\supp \widehat{\tilde{\varphi}^\perp_i}\Subset B_{\alpha}^c \\
     &x^\perp\in \chi_{B_{vt}} \implies \frac{\|x^\perp\|}{t}\leq v<2\alpha\implies
     \frac{x^\perp}{t}\not \in B_{2\alpha}^c=\{2\xi\mid \xi \in  \supp \widehat{\tilde{\varphi}^\perp_i}\}
 \end{align*}
 we may apply non-stationary phase as in the proof of Claim \ref{FreeInter} to get that for any $\ell>0$
 \begin{align*}
    \|\chi_{S_{vt}}e^{-itH_0}\tilde{\varphi}\|^2\leq C \int\limits_{B_{vt}}(1+\|x\|+t)^{-\ell}\,dx
\end{align*}
where $C$ does not depend $t$. In particular, for any $\ell$ large enough
\begin{align*}
    \|\chi_{S_{vt}}e^{-itH_0}\tilde{\varphi}\|^2\leq C(1+t)^{-\ell+d}\xrightarrow{t\rightarrow\infty }0
\end{align*}
So we can conclude that 
\begin{align*}
    \|\chi_{S_{vt}}e^{-itH}\varphi\|\xrightarrow{t\rightarrow\infty }0 
\end{align*}
Applying this to inequality (\ref{ineq:Hsur2}) combined with equation (\ref{lim:Hsur2}) we conclude that
\begin{align*}
    \braket{\psi,\varphi}=0
\end{align*}
which completes the proof.
\end{proof}
This proposition with Lemma \ref{SemiComplete} prove part \ref{Complete} of Theorem \ref{thms}, or in other words asymptotic completeness.
\end{proof}
\section{Examples}\label{ExamplesSection}
Having established our main theorem, we analyze a few special cases to see some of the variety of surface states that may occur. For this purpose, it will be convenient to work with the sufficient condition for being a surface state given in the following proposition.
\begin{proposition}\label{Hsur'}
In the notation of Section \ref{intro}
\begin{align*}
    \calH_\mathrm{sur}'\subset \calH_\mathrm{sur}
\end{align*}
\end{proposition}
\begin{proof}
Recall the definition of $\calH'_{\textrm{sur}}$:
\begin{align*}
    \calH_\textrm{sur}'(H)=\{\psi\mid\lim_{R\rightarrow\infty} \sup_{t\geq 0} \|\chi_{S_R^c}e^{-itH}\psi\|=0\}
\end{align*}
We note that for any $v>0$, $\psi \in \calH$, and $t>0$ we have
\begin{align*}
    \|\chi_{S_{vt}^c}e^{-itH}\psi\|\leq \sup_{\tau\geq 0} \|\chi_{S_{vt}^c}e^{-i\tau H}\psi\|
\end{align*}
Since this is true for any $t>0$ we can take $\lim\limits_{t\rightarrow\infty}$ on both side to get
\begin{align*}
    \lim\limits_{t\rightarrow\infty}\|\chi_{S_{vt}^c}e^{-itH}\psi\|\leq \lim\limits_{t\rightarrow\infty} \sup_{\tau\geq 0} \|\chi_{S_{vt}^c}e^{-i\tau H}\psi\|=\lim\limits_{R\rightarrow\infty} \sup_{\tau\geq 0} \|\chi_{S_{R}^c}e^{-i\tau H}\psi\|
\end{align*}
So if $\psi \in \calH_\textrm{sur}'$, the last term is $0$, and therefore $\psi \in \calH_{\textrm{sur}}$, as needed.

\end{proof}
\subsection{Surface States in $\sigma_{\mathrm{c}}(H)$}
While it is clear that eigenfunctions of $H$ are in $\calH_{\textrm{sur}}'$, and so from the above proposition are surface states, it is natural to ask whether there may also be surface states in the continuous subspace. We answer this in the affirmative via a simple example.\par
Let $d=2$ and consider a potential which depends on the $x$ coordinate only:
\begin{align*}
    &V(x,y)=V_0(x)\\
    &\supp V_0 \subset \{|x|<1\}
\end{align*}
Then we may write
\begin{align*}
    H:=-\frac{\partial^2}{\partial x^2}-\frac{\partial^2}{\partial y^2} +V(x,y)=H_x\otimes\id+\id\otimes H_y
\end{align*}
where $H_x$ and $H_y$ are the one-dimensional operators
\begin{align*}
    &H_x=-\frac{d^2}{dx^2} +V_0(x)\\
    &H_y=-\frac{d^2}{dy^2}
\end{align*}
Assume that $H_x$ has an eigenvalue $E_0$ with corresponding eigenfunction $\psi_0$. For any $\psi_1(y)\in L^2(\bbR)$, we claim that
\begin{align}
    \varphi(x,y):=\psi_0(x)\psi_1(y)
\end{align}
is in $\calH_{\textrm{sur}}(H)$.\par
To see this, note that since $H=H_x\otimes\id+\id\otimes H_y$ we may write
\begin{align*}
    e^{-itH}\varphi=e^{-itH_x}\psi_0\otimes e^{-itH_y}\psi_1=e^{-itE_0}\psi_0\otimes e^{-itH_y}\psi_1
\end{align*}
so that for all $t$
\begin{align*}
    &\|\chi_{S_n^c} e^{-itH} \varphi\|^2=\int\limits_{\bbR} \int\limits_{|x|>n}|e^{-itE_0}\psi_0(x)e^{-itH_y}\psi_1(y)|^2\,dx\,dy\\
    &=\int\limits_{\bbR} |e^{-itH_y}\psi_1(y)|^2\,dy\int\limits_{|x|>n}|e^{-itE_0}\psi_0(x)|^2\,dx\\
    &=\|e^{-itH_y}\psi_1\|^2 \int\limits_{|x|>n}|\psi_0(x)|^2dx \xrightarrow{n\rightarrow \infty }0
\end{align*}
 Therefore, by Proposition \ref{Hsur'} we conclude that $\varphi\in \calH_{\textrm{sur}}$.\par
Furthermore, if $\psi_1 \in \calH_{\textrm{ac}}(H_y)$, as $H_y$ is purely ac, we can guarantee that $\varphi\in \calH_\textrm{ac}(H)$. This is because for self-adjoint operators of the form $D=A\otimes \id+\id\otimes B$, the spectral measure of $f(x)g(y)$ with respect to $D$ is given by the convolution of the spectral measure of $f$ with respect to $A$ with the spectral measure of $g$ with respect to $B$ (see \cite{fox1975spectral} for more details).\par
\begin{remark}
In \cite{Richard}, Richard generalized this example by introducing  a class of ``Cartesian potentials'' that, roughly speaking, attain different limits in different coordinate directions. For instance, we may consider potentials of the form $V(x,y)=V_0(x)V_1(y)$, where $V_0(x)$ is as above and $V_1(y)$ decays to a limit in a short-range way: there exists some $c\in\bbR$ such that 
\begin{align*}
    \|\chi(|y|>R)(V_1(y)-c)\|_\textrm{op}\in L^1(R)
\end{align*}
Writing \begin{align*}
    H_1&=-\Delta +cV_0(x)\\
    H_x&=-\frac{d^2}{dx^2}+cV_0(x)
\end{align*}
one may infer from Theorem 1.2 in \cite{Richard} that
\begin{align*}
    \calH=\textrm{Ran}(\Omega ^-)\oplus \calH_{\textrm{pp}}(H) \oplus \textrm{Ran}(\tilde{\Omega}^-)
\end{align*}
where
\begin{align*}
    \tilde{\Omega}^-=\slim\limits_{t\rightarrow\infty }e^{itH}e^{-itH_1}(\id \otimes P_{\calH_{\textrm{pp}}}(H_x))
\end{align*}
By an argument similar to the one given for the above example, it is easy to see that $\textrm{Ran}(\tilde{\Omega}^-)\subset \calH_{\textrm{sur}}'$ so that $\calH_{\textrm{sur}}'= \calH_{\textrm{sur}}$. 
\end{remark}

\subsection{Potentials Periodic in All But One Direction}
Now suppose that $k=d-1$ and that $V$ is \emph{periodic in all but one direction} in that there are linearly independent vectors $a_1,\ldots a_{d-1}\subset \bbR^{d}$ such that $V(x+a_i)=V(x)$ for all $i$ and $x\in\bbR^d$. The additional structure of such potentials allows us to give a simpler characterization of the surface states. The proof below can be gleaned from the analysis of such systems in \cite{DaviesSimon}, but we include a proof for the sake of completeness. A similar proof for a different system may be found in \cite{saenz1981quantum}.
\begin{theorem}\label{periodicThm}
Suppose that $V$ is periodic in all but one direction. Then 
\begin{align*}
    \calH=\Ran \Omega^-\oplus \calH_\mathrm{sur}'
\end{align*}
In particular, $\calH_\mathrm{sur}'=\calH_\mathrm{sur}$.
\end{theorem}
\begin{proof}
Following \cite{DaviesSimon}, there exists
$U:\calH\rightarrow\int^\oplus_\bbT\calH(\theta)\,d\theta$ unitary that partially diagonalizes $H$. Here, $\bbT:=[0,2\pi)^{d-1}$ and $\calH(\theta)= L^2(D)$ for
\begin{align*}
    D=\{x\in\bbR^{d-1}\mid x=\sum_{i=1}^{d-1} a_iy_i\text{ for }y\in [0,1)^{d-1}\}\times \bbR
\end{align*} is the cylinder over the basic cell of the periods. For each $\theta$, we let $H_0(\theta)$ be $-\Delta$ on $\calH(\theta)$ with core given by $\psi\in L^2(D)$ with smooth extensions to $\bbR^d$ satisfying $\psi(x+a_j)=e^{i\theta_j}\psi(x)$ for all $j$ and $x\in \bbR^d$. Letting $H(\theta)=H_0(\theta)+V$, we have the unitary equivalence
\begin{align*}
    UHU^*=\int\limits_\bbT^\oplus H(\theta)\,d\theta
\end{align*}
These properties of the direct integral decomposition for periodic operators are enough to prove Theorem \ref{periodicThm}. We refer the interested reader to \cite{RSVol4} for more details about this decomposition.\par
From Theorem 5.2 of \cite{SimonEnns}, for all $\theta\in\bbT$ the wave operators 
\begin{align*}
    \Omega^{\pm}(\theta):=\slim_{t\rightarrow \mp\infty}e^{itH(\theta)}e^{-itH_0(\theta)}P_{\textrm{ac}}(H_0(\theta))
\end{align*}
exist and are complete in the sense that
\begin{align*}
    \Ran \Omega^+(\theta)=\Ran\Omega^-(\theta)=\calH_\textrm{ac}(H(\theta))
\end{align*}
and $H(\theta)$ has no singular continuous spectrum. Therefore, for each $\theta$, 
\begin{align*}
    \calH(\theta)=\Ran\Omega^-(\theta)\oplus \calH_{\textrm{pp}}(H(\theta))
\end{align*}
so that
\begin{align*}
    \calH=U^*\int\limits_\bbT^\oplus \Ran\Omega^-(\theta)\,d\theta \oplus \calH_\textrm{s}
\end{align*}
where $\calH_\textrm{s}:=\int\limits_\bbT^\oplus\calH_{\textrm{pp}}(H(\theta))\,d\theta$. These direct integrals are well-defined because $\theta\mapsto \Omega^-(\theta)$ and $\theta\mapsto P_{\textrm{pp}}(H(\theta))$ are measurable - see the Appendix to \cite{DaviesSimon}.\par
Following the proof of Theorem 1.8 in \cite{frank2003scattering}, Theorem XII.85 of \cite{RSVol4} implies that
\begin{align*}
    &Ue^{-itH}U^*=\int\limits_\bbT^\oplus e^{-itH(\theta)}\,d\theta &&
     Ue^{-itH_0}U^*=\int\limits_\bbT^\oplus e^{-itH_0(\theta)}\,d\theta
\end{align*}
Thus, for any $\psi\in L^2(\bbR^d)$
\begin{align*}
    &\|Ue^{itH}e^{-itH_0}\psi-\int\limits_\bbT^\oplus\Omega^\pm(\theta)U\psi\|^2=\\
    &\int\limits_\bbT\int\limits_D\|e^{itH(\theta)}e^{-itH_0(\theta)}(U\psi)(\theta,x)-\Omega^\pm(\theta)(U\psi)(\theta,x)\|^2\,dx\,d\theta
\end{align*}
The inner integral goes to $0$ as $t\rightarrow\pm\infty$ since $\Omega^\pm(\theta)$ exists so that by the dominated convergence theorem, we see that $\Omega^\pm=U\int\limits_\bbT^\oplus \Omega^\pm(\theta)\,d\theta U^*$. It follows that $\calH=\Ran\Omega^-\oplus \calH_\textrm{s}$. Furthermore, Proposition 6.1 of \cite{DaviesSimon} shows that $\calH_\textrm{s}\subset \calH_\textrm{sur}'$ and it is clear that $\Ran\Omega^-\subset (\calH'_\textrm{sur})^\perp$ because from part \ref{Complete} of Theorem \ref{thms} we have $\Ran \Omega^-=(\calH_\textrm{sur})^\perp$ and $\calH_{\textrm{sur}}'\subset \calH_{\textrm{sur}}$ from Proposition \ref{Hsur'}. Therefore,
\begin{align*}
    \calH=\Ran\Omega^-\oplus \calH_\textrm{s}\subset (\calH_\textrm{sur}')^\perp\oplus \calH_\textrm{sur}'
\end{align*}
which is only possible if in fact $\calH=\Ran\Omega^-\oplus \calH_\textrm{sur}' $. Since we have proven that in general $\calH_\textrm{sur}$ is the orthogonal complement of $\Ran\Omega^-$, we see that $\calH_\textrm{sur}'=\calH_\textrm{sur}$.
\end{proof}

\subsection{Transient surface states}
In this section, we exhibit a potential that induces states in $\calH_\textrm{sur}\setminus \calH_\textrm{sur}'$. Furthermore we show that one can build a potential with states that propagate in the transverse direction arbitrarily slowly in a sense specified below. Potentials of this class were originally considered by Yafaev \cite{yafaev1979break}.\par
For $d=2$ and $k=1$, let
\begin{align*}
    V(x,y)&=\Span{y}^{-2\alpha}V_0(\Span{y}^{-\alpha}x)\\
    V_0(x)&=-\chi_{[-1,1]}(x)
\end{align*}
for some $0<\alpha<\frac{1}{2}$. By writing 
\begin{align*}
    V(x,y)=-\Span{y}^{-2\alpha}\chi_{\{|x|<\Span{y}^\alpha\}}(x,y)
\end{align*}
it is clear that for any fixed $x$
\begin{align*}
    \sup\limits_{y}|V(x,y)|=\begin{cases}|x|^{-2}& |x|>1\\ 1&|x|<1 \end{cases}
\end{align*}
 Therefore,
 \begin{align*}
     \|\chi_{S_r^c}V\|=r^{-2}\in L^1(r)
 \end{align*}
i.e. the potential V satisfies (\ref{shortrangeCond}) and thus Theorem \ref{thms} applies.
\begin{remark}
One may also construct examples of potentials \emph{supported} inside a strip for which $\calH_\textrm{sur}\setminus \calH_\textrm{sur}'\neq \emptyset$. However, we consider the above example for the sake of computational simplicity.
\end{remark}
Let $h(y)$ be the operator on $L_x^2(\bbR)$ given by
\begin{align*}
    h(y)=-\frac{d^2}{dx^2}+V(x,y)
\end{align*}
Solving directly, we find that for some $E<0$, there is a normalized $\varphi_0(x)$ such that
\begin{align*}
    h(0)\varphi_0=E\varphi_0
\end{align*}
and
\begin{align*}
    \varphi_0(x)=Ce^{-c|x|}\text{ for }|x|\geq 1 
\end{align*}
By rescaling, we see that for all $y\in\bbR$
\begin{align*}
    h(y)\psi(x,y)=\Span{y}^{-2\alpha}E\psi(x,y)
\end{align*}
where
\begin{align*}
   \psi(x,y)=\Span{y}^{-\frac{\alpha}{2}}\varphi_0(\Span{y}^{-\alpha}x)
\end{align*}
Define
\begin{align*}
    J:L^2_y(\bbR)\rightarrow L^2(\bbR^2)\\
    Jf=\psi(x,y)f(y)
\end{align*}
By Theorem 15.1 in \cite{yafaev2007scattering}, since $\alpha<\frac{1}{2}$ and because $\varphi_0(x)$ clearly satisfies
\begin{align*}
    \int\limits_\bbR (1+|x|^4)|\frac{d^k}{dx^k}\varphi_0(x)|^2\,dx<\infty
\end{align*}
for all $k\leq 2$, there exists a phase function $\Xi(y,t):\bbR^2\rightarrow \bbR$ such that the modified wave operator
\begin{align*}
    \tilde{\Omega}=\lim\limits_{t\rightarrow\infty}e^{itH}JU_0(t)
\end{align*}
exists for all $f\in L^2_y(\bbR)$ where
\begin{align*}
    U_0(t)f=e^{i\Xi(y,t)}(2it)^{-\frac{1}{2}}\hat{f}(\frac{y}{2t})
\end{align*}
Moreover, $\Ran\tilde{\Omega}$ is orthogonal to $\Ran \Omega^-$ and therefore lies in $\calH_\textrm{sur}$.\par
To specify the the space distribution of states in $\Ran(\tilde{\Omega})$, for $\beta>0$ we let
\begin{align*}
    \calH_{\textrm{sur},\beta}=\{\phi\in\calH\mid \lim_{t\rightarrow\infty}\|\chi_{S_{t^\beta}^c}e^{-itH}\phi\|=0\}
\end{align*}
Intuitively, if $\phi\in \calH_{\textrm{sur},\beta}$ then at time $t$ it is localized within a strip of width $t^\beta$.
\begin{proposition}
Suppose that $\phi\in \Ran(\tilde{\Omega})$ for $\phi\neq 0$. Then $\phi\in \calH_{\mathrm{sur},\beta}$ for all $\beta>\alpha$ but not for $\beta\leq \alpha$. Moreover, $\phi\in \calH_{\mathrm{sur}}\setminus\calH_{\mathrm{sur}}'$.
\end{proposition}
\begin{remark}
The above proposition says that states in $\Ran(\tilde{\Omega})$ are localized at time $t$ in a strip of width $t^{\alpha+\epsilon}$ for any $\epsilon>0$, but not in a strip of width $t^{\alpha}$. In other words, such states propagate in the transverse direction at rate proportional to $t^\alpha$. Thus, by modulating the decay of $V$ in the longitudinal direction, choosing $\alpha$, one can create states that propagate in the transverse direction arbitrarily slowly.
\end{remark}
\begin{proof}
For $\phi\in \Ran(\tilde{\Omega})$, there exists some $f\in L^2_y(\bbR)$ such that
\begin{align*}
\lim_{t\rightarrow\infty}\|e^{-itH}\phi-JU_0(t)f\|=0
\end{align*}
so it suffices to show that
\begin{align*}
\lim_{t\rightarrow\infty}\|\chi_{S_{t^\beta}^c}JU_0(t)f\|=0
\end{align*}
for $\beta>\alpha$ and
\begin{align*}
    \lim_{t\rightarrow\infty}\|\chi_{S_{t^\beta}^c}e^{-itH}\phi\|\neq 0
\end{align*}
for $\beta\leq \alpha $. To see this, note that
\begin{align*}
    \|\chi_{S_r^c}JU_0f\|^2&=\int\limits_{|x|>r}\int\limits_\bbR |\psi(x,y)|^2(2t)^{-1}|\hat{f}(\frac{y}{2t})|^2\,dy\,dx=\int\limits_{|x|>r}\int\limits_\bbR |\psi(x,2ty)|^2|\hat{f}(y)|^2\,dy\,dx\\
    &=\int
\limits_{|x|>r}\int\limits_\bbR\Span{2ty}^{-\alpha}|\varphi_0(\Span{2ty}^{-\alpha}x)|^2|\hat{f}(y)|^2\,dy\,dx=\int\limits_\bbR\int\limits_{|x|>r\Span{2ty}^{-\alpha}}|\varphi_0(x)|^2|\hat{f}(y)|^2\,dx\,dy
\end{align*}
so we have shown that
\begin{align*}
    \|\chi_{S_r^c}JU_0f\|^2=\int\limits_\bbR g(r\Span{2ty}^{-\alpha})|\hat{f}(y)|^2\,dy
\end{align*}
 where
\begin{align*}
    g(y)=\int\limits_{|x|>|y|}|\varphi_0(x)|^2\,dx
\end{align*}
Clearly $g(0)=1$, $g(\infty)=0$, and $g(y)\geq 0$ for all $y$. By taking $r=t^\beta$ for some $\beta>0$, we see that
\begin{align*}
    \|\chi_{S^c_{t^\beta}}JU_0f\|^2=\int\limits_\bbR g(t^\beta\Span{2ty}^{-\alpha})|\hat{f}(y)|^2\,dy
\end{align*}
Given this identity, by the dominated convergence theorem we need only take the limit as $t\rightarrow\infty$ under the integral for different values of $\beta$.
For $\beta>\alpha$, this integrand goes to $0$ pointwise as $t\rightarrow\infty$ so we see that
\begin{align}\label{est:leak1}
    \lim_{t\rightarrow\infty}\|\chi_{S^c_{t^\beta}}JU_0f\|=0    
\end{align}
Conversely, for $\beta<\alpha$, the integrand goes pointwise to $g(0)|\hat{f}(y)|^2$ and for $\alpha=\beta$ to $g(\abs{2y}^{-\alpha})|\hat{f}(y)|^2$, both of which integrate to a positive quantity i.e.
\begin{align*}
    \lim_{t\rightarrow\infty}\|\chi_{S^c_{t^\beta}}JU_0f\|>0    
\end{align*}
Finally, by choosing $0<\beta<\alpha$, we see that
\begin{align*}
    \lim_{R\rightarrow\infty}\sup_{t\geq0}\|\chi_{S_{R}^c}e^{-itH}\phi\|\geq \lim_{R\rightarrow\infty}\|\chi_{S_{R}^c}e^{-iR^\frac{1}{\beta}H}\phi\|=\|\hat{f}\|
\end{align*}
by the above computation. Thus, if $\phi\neq 0$, it is not contained in $\calH_\textrm{sur}'$.
\end{proof}
\subsection{Small surface perturbations}
For a potential that is small enough in the appropriate sense, one would expect that there should be no non-trivial surface states, as is the case for $H_0$. Indeed, this holds for $k\geq 3$  from a result in \cite{boutet1996some}:
\begin{theorem} [Cor. 2.1 from \cite{boutet1996some}]
 For $V$ $\Delta$-bounded with relative bound less than one, assume that there exists some constants $C\leq \frac{(k-2)^2}{2}$ and $C'>0$ such that
\begin{enumerate}
    \item $|D^\perp V (x)|\leq \frac{C}{\|x^{\perp}\|^2}$.\label{Condition}
    \item $|D^\perp D^{\perp} V (x)|\leq \frac{C}{\|x^{\perp}\|^2}$
    \item $|V(x)|\leq \frac{C'}{|x^{\perp}|^2}$
     \item $\|D^{\perp} V \|_{\calH^2\rightarrow \calH}<\infty$
\end{enumerate}
where $D^\perp =\sum_{j=1}^k x_j\frac{\partial}{\partial x_j}$ and $\calH^2$ is the Sobolev space of order two. Then the wave operators $\Omega^\pm$ exist and define a unitary equivalence between $H$ and $H_0$.
\end{theorem}
The condition (\ref{Condition}) implies that outside of a compact neighborhood of the origin, $V(x)$ must be bounded by some dimensional constant. Therefore, the above conditions may be regarded as imposing some sort of smallness on $V$.

\subsection{Random surface potentials}
In this section, we summarize some results from \cite{de2003dynamical} which show that almost surely $\calH_\textrm{sur}$ is infinite dimensional for certain classes of random surface potentials. To this end, let
\begin{align*}
    H(\omega)=H_0+V_\omega
\end{align*}
be the random operator on $\bbR^d$ given by the potential
\begin{align*} 
    V_\omega=\sum_{k\in\bbZ^\nu}q_k(\omega)f(x-(k,0))
\end{align*}
where $f$, the single site potential satisfies
\begin{enumerate}
    \item $f\geq0$ and $f>\sigma>0$ on some non-empty open set.
    \item $f\in L^p(\bbR^d)$ for $p\geq 2$ if $d\leq 3$ and $p>\frac{d}{2}$ if $d>3$.
\end{enumerate}
and the random coefficients $q_k$ satisfy
\begin{enumerate}
    \item The $q_k(\omega)$ are i.i.d. random variables with distribution given by a measure $\mu$ such that $\supp\mu=[q_\textrm{min},0]$ for some $q_\textrm{min}<0$.
    \item $\mu$ is H\"{o}lder continuous.
    \item There exist $C,\tau>0$ such that for all $\varepsilon>0$
    \begin{align*}
        \mu([q_\textrm{min},q_\textrm{min}+\varepsilon])\leq C\varepsilon^\tau
    \end{align*}
\end{enumerate}
One can show that almost surely $\sigma(H(\omega))=[E_0,\infty)$ where
\begin{align*}
    E_0=\inf \sigma(H_0+q_\textrm{min}\sum_{k\in\bbZ^\nu}f(x-(k,0)))
\end{align*}
which is negative. Under these assumptions we have that
\begin{theorem}[Theorem 1.2 in \cite{de2003dynamical}]
For $H(\omega)$ as above, there exists $\varepsilon>0$ such that the spectrum of $H(\omega)$ is almost surely pure point in the interval $[E_0,E_0+\varepsilon]$.
\end{theorem}
Because eigenfunctions are clearly surface states, for instance by Proposition \ref{Hsur'}, this demonstrates that random models can induce an infinite dimensional space of surface states.

\appendix
\section{Properties of Phase Space Observables}\label{DaviesProperties}
In this appendix we prove several properties of the phase space observables $P_\delta(E)$ that we use above. We recall that we choose $\eta\in\calS(\bbR^d)$, such that $\|\eta\|=1$ and $\supp \hat{\eta}\subset B_1$, and $\eta=\eta^\parallel\otimes \eta^\perp$. Let $\eta_\delta$ be such that $\hat{\eta}_\delta(p)=\delta^{-\frac{d}{2}}\hat{\eta}(\frac{p}{\delta})$, a rescaling of $\eta$, so that $\supp \hat{\eta}^d_\delta\subset B_\delta$ and $\|\eta_{\delta}\|=1$.\par
Now define the following family of coherent states by translating $\eta_\delta$ in phase space:
\begin{align*}
    &\hat{\eta}_{x,p;\delta}(\xi)=e^{-ix\xi}\hat{\eta}_\delta(\xi-p)
\end{align*}
or equivalently
\begin{align*}
    &\eta_{x,p;\delta}(y)=e^{ip(y-x)}\eta_\delta(y-x)
\end{align*}
We use this to define a family, depending on $\delta>0$, of positive-operator-valued measures as in \cite{davies1980enss}, which serve as phase space observables. For any $E\subset \bbR^{2d}$ Borel and $\psi \in \calH$ let
\begin{align*}
    P_\delta(E)\psi=(2\pi)^{-d}\iint\limits_E \braket{\eta_{x,p;\delta},\psi} \eta_{x,p;\delta} \,dx\,dp
\end{align*}
\begin{proposition}\label{MomentumInt}
We have the following equality:
\begin{align*}
    (2\pi)^{-d}\int\limits_{\bbR^d}|\braket{\eta_{x,p;\delta},\psi}|^2 \,dp=\int\limits_{\bbR^d} |\overline{\eta}_\delta(y-x)\psi(y)|^2\,dy
\end{align*}
\end{proposition}
\begin{proof}
If we denote by $\calF(\cdot)$ the Fourier transform then
\begin{align*}
    (2\pi)^{-\frac{d}{2}}\braket{\eta_{x,p;\delta},\psi}&=(2\pi)^{-\frac{d}{2}}\int\limits_{\bbR^d} e^{-ip(y- x)} \bar{\eta}_\delta(y-x)\psi(y)\,dy\\
    &=e^{ipx}\calF(\bar{\eta}_\delta(\cdot-x)\psi(\cdot))(p)
\end{align*}
So, the proposition follows directly from Plancherel:
\begin{align*}
    (2\pi)^{-d}\int\limits_{\bbR^d}|\braket{\eta_{x,p;\delta},\psi}|^2 \,dp&=\int\limits_{\bbR^d}|e^{ipx}\calF(\bar{\eta}_\delta(\cdot-x)\psi(\cdot))(p)|^2 \,dp\\
    &=\int\limits_{\bbR^d}|\bar{\eta}_\delta(y-x)\psi(y)|^2\,dy
\end{align*}
as needed.
\end{proof}
\begin{corollary}
For any $\delta>0$
\begin{align*}
    P_\delta(\bbR^{2d})=\id
\end{align*}
\end{corollary}
\begin{proof}
This is a direct application of the above:
\begin{align*}
    \braket{\psi,P_\delta(\bbR^{2d})\psi}&=(2\pi)^{-d}\iiiint\limits_{\bbR^{2d}}|\braket{\eta_{x,p;\delta},\psi}|^2\,dx^\parallel\,dp^\parallel\, dx^\perp \, dp ^\perp\\
    &=\iiiint\limits_{\bbR^{2d}}|\overline{\eta}_\delta(y-x)\psi(y)|^2\,dx^\parallel\,dy^\parallel\, dx^\perp \, dy ^\perp=\|\eta_\delta\|^2\|\psi\|^2=\|\psi\|^2
\end{align*}
from which it follows that $P_\delta(\bbR^{2d})=\id$ because a self-adjoint operator is determined by its diagonal matrix elements.
\end{proof}
\begin{corollary}\label{SpaceLoc}
For $A^\parallel\subset \bbR^k, A^\perp \subset \bbR^{d-k}$ let $E=A^\parallel \times\bbR^k\times A^\perp \times\bbR^{d-k}$, and $A= A^\parallel\times A^\perp$. Then for any $\delta>0$ and $\psi\in \calS$
\begin{align*}
   (P_\delta(E)\psi)(y)=[(|\eta_\delta|^2*\chi_{A})\psi](y)
\end{align*}
\end{corollary}
\begin{proof}
This is also a direct application, where we used the short hand $x=(x^\parallel,x^\perp),p=(p^\parallel,p^\perp)$:
\begin{align*}
    \braket{\psi,P_\delta(E)\psi}&=\int\limits_{A^\parallel}\int\limits_{\bbR^k}\int\limits_{A^\perp}\int\limits_{\bbR^{d-k}}|\braket{\eta_{x,p;\delta},\psi}|^2\,dx^\parallel\,dp^\parallel\, dx^\perp \, dp ^\perp=\int\limits_A\int\limits_{\bbR^d}|\overline{\eta}_\delta(y-x)\psi(y)|^2\,dy\,dx\\
    &=\int\limits_{\bbR^d}(|\eta_\delta|^2*\chi_A)(y)|\psi(y)|^2\,dy
\end{align*}
from which the claim follows.
\end{proof}
\begin{proposition}\label{OpNormAndPos}
For each $E\subset \bbR^{2d}$ Borel, $0\leq P_\delta(E)\leq \id$. In particular
\begin{align*}
    \|P_\delta(E)\|_{\mathrm{op}}\leq 1
\end{align*}
\begin{proof}
\begin{align*}
    0\leq \braket{\psi,P_\delta(E)\psi}&=(2\pi)^{-d}\iiiint\limits_E|\braket{\eta_{x,p;\delta},\psi}|^2\,dx^\parallel\,dp^\parallel\, dx^\perp \, dp ^\perp\\
    &\leq (2\pi)^{-d}\iiiint\limits_{\bbR^{2d}}|\braket{\eta_{x,p;\delta},\psi}|^2\,dx^\parallel\,dp^\parallel\, dx^\perp \, dp ^\perp=\|\psi\|^2
\end{align*}
The operator norm bound comes from the fact that for a self-adjoint operator $A$
\begin{align*}
    \|A\|_\textrm{op}=\sup_{\|\psi\|=1}|\braket{\psi,A\psi}|
\end{align*}
from which the claim is immediate. 
\end{proof}
\end{proposition}
Next we want to be able to bound the operator norm of $AP_\delta(E)$ for another operator $A$:
\begin{proposition}\label{PNormBound}
For any $\delta>0$, and $A$ any operator we have, and any Borel set $E\subset \bbR^{2d}$:
\begin{align*}
    \|AP_\delta(E)\|_{\mathrm{op}}^2\leq (2\pi)^{-d}\iint\limits_E \|A\eta_{x,p;\delta}\|^2\,dx^\parallel\,dp^\parallel\, dx^\perp \, dp ^\perp
\end{align*}
\end{proposition}
\begin{proof}
\begin{align*}
    &\|AP_\delta(E)\psi\|^2=(2\pi)^{-2d}\int\limits_{\bbR^d}\left|\iint\limits_E\braket{\eta_{x,p;\delta},\psi}A\eta_{x,p;\delta}(y)\,dx^\parallel\,dp^\parallel\, dx^\perp \, dp ^\perp\right|^2\,dy\\
    &\leq (2\pi)^{-2d}\int\limits_{\bbR^d}\iint\limits_{\bbR^{2d}}|\braket{\eta_{x,p;\delta},\psi}|^2\,dx^\parallel\,dp^\parallel\, dx^\perp \, dp ^\perp\iint\limits_E|A\eta_{x,p;\delta}(y)|^2\,dx^\parallel\,dp^\parallel\, dx^\perp \, dp ^\perp\, dy\\
    &=(2\pi)^{-d}\|\psi\|^2\iint\limits_E \|A\eta_{x,p;\delta}\|^2\,dx^\parallel\,dp^\parallel\, dx^\perp \, dp ^\perp
\end{align*}
as needed.
\end{proof}

\begin{proposition}\label{psupport}
Let $\psi \in \calH$ be such that $\supp \hat{\psi}\subset D^\parallel\times D^\perp =D$ and let $E\subset \bbR^k\times B^\parallel\times \bbR^{d-k}\times B^\perp$ Borel where $B=B^\parallel\times B^\perp\subset \bbR^d$ satisfies $d(D,B)\geq\delta$. Then
\begin{align*}
    P_\delta(E)\psi =0
\end{align*}
Furthermore, if $F\subset\bbR^k\times D^\parallel\times \bbR^{d-k}\times D^\perp$ then
\begin{align*}
    P_{\frac{\delta}{2}}(E)P_{\frac{\delta}{2}}(F)=0
\end{align*}
\end{proposition}
\begin{proof}
The first equality follows directly from the fact that
\begin{align*}
    \braket{\eta_{x,p;\delta},\psi}=\int\limits_{\bbR^d} e^{ix\xi}\bar{\hat{\eta}}_\delta (\xi-p)\hat{\psi}(\xi)\,d\xi=0
\end{align*}
for $p\in B$ since $\supp \hat{\eta}_{x,p;\delta}\subset B_\delta+p$.\par
Similarly, the second equality comes from the fact that for any $\varphi \in \calH$ 
\begin{align*}
    \supp \widehat{P_{\frac{\delta}{2}}(F)\varphi}\subset D+B_{\frac{\delta}{2}}
\end{align*}
and an application of the first equality. 
\end{proof}
\begin{proposition}\label{BoundedByIndicator}
For any $\delta>0$, and for any Borel set $D\subset \bbR^d$, suppose that\\
${E\subset D^\parallel\times \bbR^k \times D^\perp\times \bbR^{d-k}}$ is a Borel set, and denote $D=D^\parallel\times D^\perp$. Then for any $\varphi \in \calH$
\begin{align}\label{eq:bound}
    \|P_{\delta}(E)\varphi\|^2\leq \|(|\eta_\delta|^2*\chi_{D})\varphi\|\|\varphi\|
\end{align}
\end{proposition}
\begin{proof}
The inequality (\ref{eq:bound}) is a result of the fact that $P_\delta^2(E)\leq P_\delta(E)$ (which is easy to establish since $0\leq P_\delta(E)\leq\id$):
\begin{align*}
    \|P_\delta(E)\varphi\|^2&= \braket{P_\delta^2(E)\varphi,\varphi}\leq \braket{P_\delta(E)\varphi,\varphi}\leq \braket{P_\delta( D^\parallel\times \bbR^k\times D^\perp\times \bbR^{d-k})\varphi,\varphi}\\
    &=\braket{(|\eta_\delta|^2*\chi_{D})\varphi,\varphi}\leq \|(|\eta_\delta|^2*\chi_{D})\varphi\|\|\varphi\|
\end{align*}
as needed.
\end{proof}
For the following claims, suppose that 
\begin{align*}
    \eta(x)=\eta^\parallel(x^\parallel)\eta^\perp(x^\perp)
\end{align*}
where $\eta^\parallel$ and $\eta^\perp$ are functions in $\calS(\bbR^k)$ and $\calS(\bbR^{d-k})$, respectively, of $L^2$ norm $1$. It is easy to see that in this case
\begin{align*}
    \eta_{x,p;\delta}(y)=\eta^\parallel_{x^\parallel,p^\parallel;\delta}(y^\parallel)\eta_{x^\perp,p^\perp;\delta}^\perp(y^\perp)
\end{align*}
where the shifted functions $\eta^\parallel_{x^\parallel,p^\parallel;\delta}(y^\parallel)$ and $\eta_{x^\perp,p^\perp;\delta}^\perp(y^\perp)$ are defined analogously to before. Furthermore, $P_\delta^\parallel$ and $P_\delta^\perp$ are defined as operators on $L^2(\bbR^k)$ and $L^2(\bbR^{d-k})$, respectively, in the obvious way.
 \begin{proposition} \label{TensorClaim}
 Under the above choice of $\eta$,
 if $E= E^\parallel\times E^\perp\subset \bbR^{2k}\times \bbR^{2(d-k)}$ then we have
\begin{align*}
    P_\delta(E)=P^\parallel_\delta(E^\parallel)\otimes P^\perp_\delta(E^\perp)
\end{align*}

 \end{proposition}
\begin{proof}
For $\psi^\parallel\in L^2(\bbR^k)$ and $\psi^\perp\in L^2(\bbR^{d-k})$
\begin{align*}
    &P^d_\delta(E)(\psi^\parallel\otimes \psi^\perp)\\
    &=(2\pi)^{-d}\iint\limits_{E^\parallel}\,\iint\limits_{E^\parallel} \braket{\eta^\parallel_{x^\parallel,p^\parallel;\delta}\otimes \eta^\perp_{x^\perp,p^\perp;\delta},\psi_k\otimes \psi_{d-k}} \eta^\parallel_{x^\parallel,p^\parallel;\delta}\otimes \eta^\perp_{x^\perp,p^\perp;\delta} \,dx^\parallel\,dx^\perp\,dp^\parallel\,dp^\perp\\
    &=  P^\parallel_\delta(E^\parallel)\psi^\parallel\otimes P^\perp_\delta(E^\perp) \psi^\perp
\end{align*}
Since $P_\delta(E)$ acts as claimed on elementary tensors, the claim is established by the definition of the tensor product of two operators. 
\end{proof}

\begin{corollary}\label{PNormBoundTensor}
For any $\delta>0$, let $A=B\otimes C$  where $B$ is an operator acting on $L^2(\bbR^k)$ and $C$ acts on $L^2(\bbR^{d-k})$. Then for $E$ of the above form
\begin{align*}
    \|AP_\delta(E)\|_{\mathrm{op}}^2\leq (2\pi)^{-d}\iint\limits_{E^\parallel} \|B\eta_{x,p;\delta}\|^2\,dx\,dp\cdot \iint\limits_{E^\perp} \|C\eta_{x,p;\delta}\|^2\,dx\,dp
\end{align*}
and 
\begin{align*}
    \|P_\delta(E)\|_{\mathrm{op}}=\| P^\parallel_\delta(E^\parallel)\|_{\mathrm{op}}\| P^\perp_\delta(E^\perp) \|_{\mathrm{op}}
\end{align*}
\end{corollary}
\begin{proof}
This is immediate from Proposition \ref{TensorClaim} and Proposition \ref{PNormBound}.
\end{proof}
\section{Potentials that Decay in $x^\perp$}\label{appendix:Decay}
In this appendix, we explain how our proofs may be adjusted to accommodate potentials satisfying
\begin{align}\label{Decay:cond}
	&\|V\chi_{S_R^c}\|_{\textrm{op}}\in L^1 \\
	&\sup_{x\in\bbR^d}|V(x)|=M<\infty
\end{align}\par
To see the existence, or part \ref{existenceTheorem} of Theorem \ref{thms}, for such potentials, we fix $\varepsilon\in (0,2\alpha)$ and change inequality (\ref{ext:PotDecay}) so that it reads
\begin{align*}
	&\|Ve^{-itH_0}\psi\|\leq M\|\psi^\parallel\|\|\chi_{B_{\varepsilon t}}e^{-itH_0^\perp}\psi^\perp\|+\|V\chi_{S_{\varepsilon t}^c}\|_{\textrm{op}}\|\psi\|
\end{align*}
The condition on $\varepsilon$ guarantees that $2\alpha >\frac{\varepsilon t}{t}$, which allows us to bound the first summand in the above by $C(1+t)^{-\ell+d}$ for any $\ell>0$ (compare to inequality (\ref{ExistEq})). This, combined with the condition (\ref{Decay:cond}), lets us conclude the existence of the wave operators.\par
For part \ref{Complete} of Theorem \ref{thms}, in the proof of Lemma \ref{omegalem} must be modified by fixing \\
$\varepsilon <\frac{1}{8}$ and replacing (\ref{eq:1}) by
\begin{align*}
    &\|(\Omega^--\id)\varphi_{n;\textrm{out}}\|\\
    &\leq M\int\limits_0^\infty \|\chi_{S_{\varepsilon (n+mt)}}e^{-itH_0} \varphi_{n;\textrm{out}}\|\,dt+\int\limits_0^\infty \|V\chi_{S_{\varepsilon (n+mt)}}\|_{\textrm{op}}\| \varphi_{n;\textrm{out}}\|\,dt \label{eq:1}
\end{align*}
Again, the second summand decays as per  condition (\ref{Decay:cond}). For the first summand, we must only change Claim \ref{geoclaim} to allow $y\in S_{\varepsilon(n+mt)}$, which is achieved via the restriction on $\varepsilon$ . Similar adjustment will give the result for Lemma \ref{VertClaim}. After this, the proof works as written. 

\bibliographystyle{amsplain}
\bibliography{bibliography}
\end{document}